\documentclass[11pt]{article} 
\usepackage{booktabs}
\usepackage{mathrsfs}
\usepackage{natbib,graphicx,setspace,lscape,longtable}
\usepackage{mathrsfs,amsmath,amsthm,amssymb,color}
\usepackage{natbib,epsfig,graphicx,pdfpages}
\usepackage{rotating}
\usepackage{subfigure}
\usepackage{subcaption}
\usepackage{graphicx}
\usepackage{float}
\usepackage{xurl}
\usepackage{threeparttable}
\usepackage{appendix}
\usepackage{tabularx} 
\usepackage{placeins}
\usepackage[margin=1in]{geometry}
\usepackage{setspace}
\usepackage[linesnumbered, ruled]{algorithm2e}

\bibpunct{(}{)}{;}{a}{,}{,}

\usepackage{hyperref} 
\hypersetup{
    colorlinks=true,
    linkcolor=blue,
    citecolor=blue,
    urlcolor=cyan,
    pdfborder={0 0 0}
}

\newtheorem{theorem}{Theorem}
\newtheorem{lemma}{Lemma}
\newtheorem{proposition}{Proposition}
\newtheorem{remark}{Remark}
\usepackage{amsthm}

\newtheorem{assumption}{{Assumption}}

\def\beq{\begin{equation}}
\def\eeq{\end{equation}}
\def\beqr{\begin{eqnarray}}
\def\eeqr{\end{eqnarray}}
\def\beqrs{\begin{eqnarray*}}
\def\eeqrs{\end{eqnarray*}}
\def\bet{\begin{theorem}}
\def\eet{\end{theorem}}
\def\bel{\begin{lemma}}
\def\eel{\end{lemma}}
\def\bep{\begin{proposition}}
\def\eep{\end{proposition}}
\def\bg{\begin{figure}[tbph]\begin{center}}
\def\eg{\end{center}\end{figure}}

\def\bc{\begin{center}}
\def\ec{\end{center}}

\def\wt{\widetilde}
\def\wh{\widehat}

\def\rank{\mbox{rank}}

\numberwithin{equation}{section}

\newcommand{\Cov}{\textnormal{Cov}}

\newcommand{\bA}{{\mathbf A}}
\newcommand{\bB}{{\mathbf B}}

\newcommand{\bH}{{\mathbf H}}
\newcommand{\bI}{{\mathbf I}}

\newcommand{\bL}{{\mathbf L}}
\newcommand{\bM}{{\mathbf M}}
\newcommand{\bQ}{{\mathbf Q}}
\newcommand{\bP}{{\mathbf P}}
\newcommand{\bR}{{\mathbf R}}

\newcommand{\bT}{{\mathbf T}}

\newcommand{\bb}{{\mathbf b}}

\newcommand{\bff}{{\mathbf f}}

\newcommand{\bu}{{\mathbf u}}
\newcommand{\bv}{{\mathbf v}}
\newcommand{\bw}{{\mathbf w}}
\newcommand{\bx}{{\mathbf x}}
\newcommand{\by}{{\mathbf y}}

\newcommand{\bfeta}  {\boldsymbol{\eta}}

\newcommand{\bSigma}{\boldsymbol{\Sigma}}

\newcommand{\bve}{\mbox{\boldmath$\varepsilon$}}

\newcommand{\bPhi} {\boldsymbol{\Phi}}

\newcommand{\bGamma} {\boldsymbol{\Gamma}}

\newcommand{\bC}{{\mathbf C}}

\newcommand{\bzero}{{\mathbf 0}}

\newcommand{\ve}{{\varepsilon}}
\renewcommand{\epsilon}{{\ve}}
\renewcommand{\hat}{\widehat}
\newcommand{\T}{{\rm T}}
\def\wt{\widetilde}
\renewcommand{\tilde}{\wt}
\newcommand{\tr}{\mbox{tr}}

\begin{document}

\title{ Structural Change Detection in High-Dimensional Transformed Factor Models via Canonical Correlation Analysis}

\author{Lei Jia, Shouri Hu, and Zhaoxing Gao\thanks{ Corresponding author: \texttt{zhaoxing.gao@uestc.edu.cn} (Z. Gao),  School of Mathematical Sciences, University of Electronic Science and Technology of China, Chengdu, 611731 P.R. China.} \\
School of Mathematical Sciences,
University of Electronic Science and Technology of China
}

\date{}

\maketitle
\begin{abstract}
This paper develops a canonical-correlation-based method for detecting structural changes in high-dimensional transformed factor models. 
The proposed approach exploits the low-rank canonical-correlation structure induced by dynamically dependent common factors, while serially uncorrelated idiosyncratic components correspond to a noise subspace with zero canonical correlations. 
We construct an eigenvalue-ratio criterion that measures residual dynamic dependence in the estimated noise subspace and identifies the true change point under sufficient separation of the regime-specific loading spaces or dynamic canonical correlation structures. 
Since the change-point location and the regime-specific factor numbers are both unknown, we further propose an alternating iterative estimation procedure that updates them sequentially until convergence. 
Under suitable mixing and moment conditions, we establish asymptotic properties of the proposed estimators, with convergence rates depending explicitly on factor strength, cross-sectional dimension, and sample size. 
Monte Carlo experiments and empirical applications to intraday stock returns and U.S. temperature series demonstrate the finite-sample performance and practical usefulness of the proposed method.
\end{abstract}


\newpage

\section{Introduction}

High-dimensional time series are now routinely encountered in economics, finance, environmental studies, and many other empirical fields, where a large number of variables are observed over time and often exhibit both temporal dynamics and strong cross-sectional dependence. Although classical statistical models provide a foundation for time-series analysis, many of them were developed primarily for low-dimensional settings. For instance, the vector autoregressive moving-average (VARMA) model is a canonical parametric framework for multivariate time-series analysis. In high-dimensional settings, however, unrestricted VARMA models quickly become overparameterized and may suffer from identification difficulties. One way to mitigate these difficulties is to impose simplified dynamic structures, such as the scalar component model (SCM) proposed by \cite{TiaoTsay_1989}, or sparsity and regularization restrictions; see, for example, \cite{ShojaieMichailidis_2010}, \cite{HanLuLiu_2015}, and \cite{Davis2016}.

An alternative and widely used strategy is dimension reduction, which reduces both the statistical and computational complexity of high-dimensional time-series modeling. Prominent dimension-reduction approaches include the canonical correlation analysis (CCA) framework of \cite{BoxTiao_1977} and principal component analysis (PCA)-based methods developed in, for example, \cite{StockWatson_2002}. Factor models have become especially influential because they represent temporal and cross-sectional dependence through a small number of unobserved common factors; see, among others, \cite{ChamberainRothschild_1983}, \cite{BaiNg_Econometrica_2002}, \cite{bai2003}, \cite{StockWatson_2005}, \cite{panyao2008}, \cite{LamYaoBathia_Biometrika_2011}, \cite{lamyao2012}, and \cite{changguoyao2015}. Within this literature, \cite{LamYaoBathia_Biometrika_2011} and \cite{lamyao2012} estimate factor loadings and latent factors through the eigen-decomposition of sample cross-autocovariance matrices. By contrast, \cite{gaotsay_2019}, building on \cite{TiaoTsay_1989}, proposed a transformed factor approach that uses CCA of time-lagged moment matrices to separate dynamically dependent common factors from serially uncorrelated idiosyncratic components. A key feature of this framework is that, after a nonsingular linear transformation of the original high-dimensional series, the dynamically dependent factor component can be distinguished from the noise component through nonzero canonical correlations between the current observations and their lagged values.

Most classical high-dimensional factor models assume that the loading matrix is time invariant. In empirical applications, however, the underlying data-generating mechanism may undergo structural changes that alter model parameters. For example, major policy adjustments, financial crises, technological changes, or shifts in market conditions may induce abrupt changes in factor loadings, thereby changing the effects of latent factors on observed variables. Ignoring such unobserved structural changes may lead to substantial estimation bias, spurious inflation in the estimated number of factors, and degraded forecasting or inferential performance. Consequently, detecting structural changes is important for reliable factor modeling and subsequent empirical analysis.

A growing literature has investigated structural changes in high-dimensional factor models; see, among others, \cite{BreEick_2011}, \cite{Chenetal_2014}, \cite{LiuChen_D_2016}, and \cite{MaSu_2018}. \cite{LiuChen_D_2016} proposed a regime-switching dynamic factor model in which regime transitions are governed by an unobserved Markov chain, so that the transition times effectively serve as change points. Subsequently, \cite{BaiHanShi_2020} developed a least-squares-based estimation framework for change-point detection in high-dimensional time series. Their method assumes that a piecewise factor model with a structural change point can be recast as an equivalent pseudo-factor model without an explicit change point. Their procedure first estimates the pseudo-factor sequence and then selects the time point that minimizes the residual sum of squares. \cite{LiuChen_D_2020} further studied a threshold-variable approach to identifying change points within the factor estimation procedure. Building on this line of research, \cite{LiuZhang_2022} proposed a data-driven projection method for sequentially detecting change points in latent factor models. Their method projects a second-order sample cross-moment matrix, constructed under a candidate partition, onto the estimated orthogonal noise space associated with the true change-point partition, and then minimizes the squared norm of the projected matrix.

Estimating the number of latent factors is another fundamental model-selection problem in high-dimensional time-series analysis. \cite{Onatski_2010} and \cite{Onatski_2012} proposed thresholding methods for determining the number of factors based on the empirical spectral distribution of the sample covariance or autocovariance matrix. \cite{lamyao2012} estimated the number of factors by maximizing the ratio of adjacent eigenvalues derived from the cross-autocovariance matrix. Although eigenvalue-ratio methods are computationally convenient, they may be sensitive to dominant factors and can suffer from a masking effect, leading to substantial underestimation of the true number of factors. To address this issue, \cite{Wu_2016} proposed applying a monotone nonlinear transformation, such as the standard normal cumulative distribution function, to the eigenvalues before forming adjacent ratios. Such a transformation reduces the scale disparity between strong and weak factors and mitigates the masking effect induced by dominant factors. Following this idea, \cite{Xia_TCR_2017} introduced the transformed contribution ratio (TCR) method, which incorporates cumulative contribution information to improve finite-sample accuracy. In the transformed factor setting, \cite{gaotsay_2019} estimated the number of factors through a sequential hypothesis test for zero canonical correlations and established the asymptotic chi-square distribution of the proposed test statistic under the null hypothesis.

In the presence of structural change, however, factor-number estimation and change-point estimation are intrinsically connected. A misspecified break location may mix observations from different regimes and thereby distort the estimated factor dimension. Conversely, an inaccurate factor-number estimate may affect the construction of the noise subspace and distort the change-point criterion. This circular dependence is particularly relevant in transformed factor models, where the dynamic signal is identified through canonical correlations and the noise space is determined by the number of zero canonical correlations. Therefore, a change-point procedure for transformed factor models should jointly account for the unknown break location and the regime-specific factor numbers.

Motivated by the transformed factor framework of \cite{gaotsay_2019} and the structural change literature for high-dimensional factor models, this paper develops a CCA-based procedure for estimating a single structural change point in high-dimensional transformed factor models. The proposed method exploits the low-rank canonical-correlation structure induced by dynamically dependent common factors. At the true change point, the regime-specific canonical correlation matrix has a signal subspace associated with the common factors and a complementary noise subspace associated with zero canonical correlations. We therefore construct an eigenvalue-based criterion that measures the residual canonical dependence in the estimated noise subspace. The population counterpart of this criterion is minimized at the true change point, which provides the basis for identification and estimation.

The contribution of this paper is threefold. First, we introduce a canonical-correlation noise-space criterion for structural change detection in transformed factor models. Unlike methods based directly on covariance or cross-moment matrices, the proposed criterion targets the dynamic dependence structure captured by canonical correlations and is therefore tailored to the transformed factor framework. Second, we develop an alternating iterative estimation (AIE) procedure to handle the mutual dependence between the unknown change point and the regime-specific factor numbers. The procedure alternates between factor-number estimation and change-point estimation until convergence, thereby reducing the sensitivity of the change-point estimator to initial factor-number choices. Third, we establish the asymptotic properties of the proposed estimators under suitable mixing and moment conditions. The resulting convergence rates explicitly reflect the effects of factor strength, cross-sectional dimension, and sample size. Monte Carlo experiments and empirical applications further illustrate the finite-sample performance and practical usefulness of the proposed method.

The remainder of this paper is organized as follows. Section \ref{sec2} introduces the single-change-point transformed factor model and develops the proposed CCA-based change-point estimation procedure. It also presents the alternating iterative estimation algorithm for jointly estimating the change point and the regime-specific factor numbers. Section \ref{sec3} establishes the theoretical properties of the proposed estimators. Section \ref{sec4} reports Monte Carlo simulation results. Section \ref{sec5} applies the proposed method to intraday stock returns and daily temperature series. Section \ref{sec6} concludes. Technical proofs are collected in the Appendix.

We conclude this section by introducing notation used throughout the paper. For a matrix $\bH=(h_{ij})$, $||\bH||_2=\sqrt{\lambda_{\max} (\bH^{T} \bH ) }$ and $||\bH||_F=\sqrt{\tr(\bH^{T} \bH)}$ denote, respectively, the $L_2$ norm and the Frobenius norm, where $\lambda_{\max} (\bH) $ is the largest eigenvalue of $\bH$, and $\tr(\cdot)$ denotes the trace operator. $\text{Diag}[\bH_{1},\bH_{2},\ldots]$ denotes a block diagonal matrix, and $\mathcal{M}(\bH)$ denotes the column space spanned by the matrix $\bH$. $\sigma_{\min}^{+}$ denotes the minimum positive singular value. The superscript $T$ denotes the transpose of a vector or matrix. In addition, we use the notation $a \asymp b$ to denote $a=O(b)$ and $b=O(a)$. Finally, $\left \lfloor x \right \rfloor $ and $\left \lceil x \right \rceil $ denote the greatest integer less than or equal to $x$ and the smallest integer greater than or equal to $x$, respectively.

\section{Methodology}\label{sec2}
\subsection{Model framework}
Let $\boldsymbol{Y}=(\by_{1},\by_{2},\ldots,\by_{T})$ be a $p$-dimensional time series, where each observation at time $t$ is a column vector $\by_{t}=(y_{1t},y_{2t},\ldots,y_{pt})^{T}$. We allow the latent factor structure to undergo a single structural change point, which partitions the sample into two temporal regimes. For any $1\le t \le T$, we consider a structural transformed factor model
with a single change point:
\begin{equation}\label{cpf-transform}
    \by_t=\left\{
    \begin{array}{ll}
    \wt\bL_{1}^{(1)}\bff_{t}^{(1)} + \wt\bL_{2}^{(1)}\bve_{t}^{(1)}, 
    & 1\le t\le k_0, \\[2mm]
    \wt\bL_{1}^{(2)}\bff_{t}^{(2)} + \wt\bL_{2}^{(2)}\bve_{t}^{(2)}, 
    & k_0 < t \le T.
    \end{array}
    \right.
\end{equation}
Here, $\bff_{t}^{(i)}\in\mathbb R^{r_i}$, $i=1,2$, denotes the latent
common factor vector in the pre- and post-change regimes, respectively, and
the positive integer $k_0$ is the true but unknown location of the change
point. Furthermore, $\bve_{t}^{(i)} \in \mathbb R^{v_i}$ denotes a serially
uncorrelated idiosyncratic component with mean zero and covariance matrix
$\Cov(\bve_{t}^{(i)})$. For $i=1,2$, the matrices
$\wt\bL_{1}^{(i)}\in \mathbb R^{p\times r_i}$ denote the unknown
regime-specific loading matrices associated with the common factors, whereas
$\wt\bL_{2}^{(i)}\in \mathbb R^{p\times v_i}$ are the corresponding loading
matrices associated with the idiosyncratic components.

When $r_i+v_i=p$ and the matrix
$[\wt\bL_{1}^{(i)},\wt\bL_{2}^{(i)}]$ is nonsingular, Model
\eqref{cpf-transform} can be regarded as a transformed factor model. In
particular, by defining
$\bT_i=[\wt\bL_{1}^{(i)},\wt\bL_{2}^{(i)}]^{-1}$, we have
$\bT_i\by_t=(
        \bff_t^{(i)}{^T},
        \bve_t^{(i)}{^T})^T$
within regime $i$. Hence, after the regime-specific linear transformation,
the observed vector is decomposed into a signal component and an innovation
component. This formulation is consistent with the scalar component model of
\cite{TiaoTsay_1989} and \cite{gaotsay_2019}. For each $i=1,2$, we assume $\bff_{t}^{(i)}$ and $\bve_{t}^{(i)}$ are mutually independent. To model the temporal dependence, we further assume that $\bff_{t}^{(i)}$ follows a stationary VAR($d$) process, 
\begin{equation}\label{var-d}
 \bff_{t}^{(i)} = \sum_{j=1}^d \bPhi_{j}^{(i)} \bff_{t-j}^{(i)} + \bu_{t}^{(i)},\quad i=1,2,
 \end{equation}
where $\{\bu_{t}^{(i)}\in\mathbb{R}^{r_i},t=1,2,\ldots\}$ is the innovation series with diagonal covariance matrix, and $\bPhi_{j}^{(i)}$ denotes the autoregressive coefficient matrix at lag $j$.
 
Model (\ref{cpf-transform}) can be equivalently reformulated into the following compact form:
\begin{equation}\label{cpf-transform-simple}
    \by_t=\left\{\begin{matrix}
    \wt\bL^{(1)}\bfeta_{t}^{(1)},\quad 1\le t\le k_0 
    \\
    \wt\bL^{(2)}\bfeta_{t}^{(2)},\quad k_0 < t \le T,
    \end{matrix}\right.
\end{equation}
where $\wt\bL^{(i)}=(\wt\bL_1^{(i)},\wt\bL_2^{(i)})$ is a  nonsingular $p \times p$ real-valued matrix, and $\bfeta_{t}^{(i)}=[\bff_{t}^{(i)^T},\bve_{t}^{(i)^T}]^{T}$ is a $p$-dimensional vector whose covariance matrix is block diagonal, namely 
\[\mathrm{Diag}[\Cov(\bff_t^{(i)}),\Cov(\bve_t^{(i)})].\] 
Note that any nonsingular linear transformation of the original series $\by_{t}$ does not alter the canonical correlation between $\by_{t}$ and its lagged values; consequently, the representation in $(\ref{cpf-transform-simple})$ is not unique. In light of this rotational non-uniqueness, we consider the following standardized transformed factor model with a change point:
\begin{equation}\label{cpf-transform-simple2}
    \bSigma_\by^{-1/2}\by_t=\left\{\begin{matrix}
    \bL^{(1)}\bfeta_{t}^{(1)},\quad 1\le t\le k_0 
    \\
    \bL^{(2)}\bfeta_{t}^{(2)},\quad k_0 < t \le T,
    \end{matrix}\right.
\end{equation}
where $\bSigma_\by=\Cov(\by_t)$ is the covariance matrix used for standardization, and $\bL^{(i)}=(\bL_{1}^{(i)},\bL_{2}^{(i)})=\bSigma_{\by}^{-1/2} \wt\bL^{(i)} $ is a $p\times p$ orthogonal matrix satisfying $\bL^{(i)}\bL^{(i)^{T}}=\bL^{(i)^{T}}\bL^{(i)}=\bI_{p}$. 
Let $\gamma_1$ and $\gamma_2$ with $0<\gamma_1<\gamma_2<1$ denote the lower and upper truncation constants for the change-point search interval, respectively, so that the effective search-window length is $n=\left \lfloor \gamma_{2}T \right \rfloor-\left \lfloor \gamma_{1}T \right \rfloor$. For any candidate change point $k$ within the interval  $ [\left \lfloor \gamma_{1}T \right \rfloor,\left \lfloor \gamma_{2}T \right \rfloor]$, let $\mathbb{I}_{t,1}(k)$ and $\mathbb{I}_{t,2}(k)$ denote the corresponding indicator functions such that $\mathbb{I}_{t,1}(k)=1$ if $1\le t \le k$ and $\mathbb{I}_{t,2}(k)=1$ if $k+1\le t \le T$. For notational convenience, we use \(\mathbb{I}_{t,i}(k)\) to indicate the segment induced by a candidate split \(k\). All covariance and cross-covariance matrices indexed by \(i\) and \(k\) are understood as segment-restricted quantities computed within the corresponding segment. Thus, for example,
\[
\Cov(\by_t\mathbb{I}_{t,i}(k))
\equiv
\Cov(\by_t\mid \mathbb{I}_{t,i}(k)=1),
\]
and the same convention applies to cross-covariance matrices such as
\(\Cov(\by_t\mathbb{I}_{t,i}(k),\by_{t,m}\mathbb{I}_{t,i}(k))\) below. Under the partition induced by the candidate change point $k$, the latent factors can be explicitly recovered as follows:
\begin{equation}\label{trans_factorepsilon}
    [\bL^{(i)}]^{-1}[\bSigma_{\by,i}(k)]^{-1/2}\by_t \mathbb{I}_{t,i}(k)=
    \left[\begin{array}{c}\bff_{t}^{(i)} \\ \bve_{t}^{(i)}\end{array}\right],\quad i=1,2,
\end{equation}  
where \(\bSigma_{\by,i}(k)=\Cov(\by_t\mathbb{I}_{t,i}(k))\) denotes the segment-restricted covariance matrix of \(\by_t\) in the \(i\)-th segment induced by the candidate split \(k\), following the notational convention above.

Let $\by_{t,m}=(\by_{t-1}^T ,\by_{t-2}^T,\ldots,\by_{t-m}^T)^T$ be the $mp$-dimensional vector of past $m$ lagged observations of $\by_t$ where the integer m satisfies $1\le m \le t-1$. We then define the cross-covariance matrix $\bSigma_{\by\by_m,i}(k)=\Cov(\by_{t}\mathbb{I}_{t,i}(k),\by_{t,m}\mathbb{I}_{t,i}(k))$ and the autocovariance matrix of the lagged vector $\bSigma_{\by_m,i}(k)=\Cov(\by_{t,m}\mathbb{I}_{t,i}(k))$. To estimate the parameters $\{\bL^{(1)},\bL^{(2)},\bff_{t}^{(1)},\bff_{t}^{(2)},k_0\}$ under model (\ref{cpf-transform-simple2}), we compare the canonical correlation structures between $\by_t$ and its lagged vector $\by_{t,m}$ across the two regimes induced by the candidate change point $k$. From a statistical perspective, applying CCA to maximize the linear association between $\by_t$ and $\by_{t,m}$ is equivalent to eigen-decomposition on the following target matrix:
\begin{equation}\label{M}
    \bM_i(k)=[\bSigma_{\by,i}(k)]^{-1/2}\bSigma_{\by\by_m,i}(k)[\bSigma_{\by_m,i}(k)]^{-1}\bSigma_{\by_m\by,i}(k)[\bSigma_{\by,i}(k)]^{-1/2},\quad i=1,2.
\end{equation}
Let $\bfeta_{t,m}^{(i)}=(\bfeta_{t-1}^{(i)\ ^T} ,\bfeta_{t-2}^{(i)\ ^T},...,\bfeta_{t-m}^{(i)\ ^T})^T$. We define $\bSigma_{\bfeta\bfeta_m,i}(k)$, $\bSigma_{\bfeta,i}(k)$ and $\bSigma_{\bfeta_m,i}(k)$ analogously as the covariance matrices of the corresponding stochastic vectors. When $k=k_0$, the matrix admits the following representation:
\begin{equation}\label{M_yita}
    \bM_i(k_0)=\bL^{(i)}\bSigma_{\bfeta\bfeta_m,i}(k_0)[\bSigma_{\bfeta_m,i}(k_0)]^{-1}\bSigma_{\bfeta_m\bfeta,i}(k_0)\bL^{(i)\ T},\quad i=1,2.
\end{equation}
Under the assumption that the idiosyncratic components are serially uncorrelated, it follows that $\Cov(\bve_{t}^{(i)}\mathbb{I}_{t,i}(k),\bfeta_{t,m}^{(i)}\mathbb{I}_{t,i}(k))=0$ for $i=1,2$. Therefore, we have the simplified form of the matrix $\bM_{i}(k_0)$:
\begin{equation}\label{M_yita_simple}
    \bM_i(k_0)=\bL_{1}^{(i)}\bSigma_{\bff\bfeta_m,i}(k_0)[\bSigma_{\bfeta_m,i}(k_0)]^{-1}\bSigma_{\bfeta_m\bff,i}(k_0)\bL_{1}^{(i)\ T},\quad i=1,2,
\end{equation}
where $\bSigma_{\bff\bfeta_m,i}(k_0)=\Cov(\bff_{t}^{(i)}\mathbb{I}_{t,i}(k_0),\bfeta_{t,m}^{(i)}\mathbb{I}_{t,i}(k_0))$.

From (\ref{M_yita_simple}), at the true change point $k_0$, the CCA matrix $\bM_i(k_0)$ is driven solely by the common factor component, while the idiosyncratic component contributes no serial dependence under the serial-uncorrelated noise assumption. As a consequence, the rank of $\bM_i(k_0)$ is determined by the factor dimension $r_i$, and the remaining eigenvalues associated with the orthogonal noise space are theoretically equal to zero. This structure provides the foundation for change-point detection: when the candidate point coincides with the true change-point location, the residual canonical dependence lying in the estimated noise subspace is minimized.

\subsection{Estimation of the change point and factor loadings}
We first consider the case in which the regime-specific factor numbers $r_i$ in both regimes are known. For each regime $i \in \{1,2\}$, let $\lambda_{j}^{(i)}(k)$ for $j=1,2,\ldots,p,$ denote the $j$-th largest eigenvalue of $\bM_i(k)$, and let $\bb_{j}^{(i)}(k)$ be the corresponding unit eigenvector. Define the signal-subspace eigenvector matrix $\bP^{(i)}(k)=(\bb_{1}^{(i)}(k),\bb_{2}^{(i)}(k),...,
\bb_{r_i}^{(i)}(k))$ and the complementary noise-subspace eigenvector matrix $\bQ^{(i)}(k)=(\bb_{r_{i}+1}^{(i)}(k),\bb_{r_{i}+2}^{(i)}(k),\ldots,\bb_{p}^{(i)}(k))$, such that $(\bP^{(i)}(k),\bQ^{(i)}(k))$ forms a complete orthogonal basis of $\mathbb{R}^p$. From (\ref{M_yita_simple}), it is straightforward to verify that the factor loading space $\mathcal{M}(\bL_1^{(i)})$ is equivalent to the eigenspace $\mathcal{M}(\bP^{(i)}(k_0))$. Therefore, strict orthogonality holds between the loading space and the noise subspace, i.e., $\mathcal{M}(\bL_1^{(i)})\ \bot \ \mathcal{M}(\bQ^{(i)}(k_0))$. For any $i=1,2$, we obtain the following relation:
\begin{equation}\label{Qt@M@Q_k0}
\begin{aligned}
    &\bQ^{(i)}(k_0)^T\bM_{i}(k_0)\bQ^{(i)}(k_0)\\
    &\quad\quad =\bQ^{(i)}(k_0)^T\bL_{1}^{(i)}\bSigma_{\bff\bfeta_m,i}(k_0)[\bSigma_{\bfeta_m,i}(k_0)]^{-1}\bSigma_{\bfeta_m\bff,i}(k_0)\bL_{1}^{(i)\ T}\bQ^{(i)}(k_0)\\
    &\quad\quad=\mathbf{O}_{v_i}.
\end{aligned}
\end{equation}
Furthermore, we have $\rank(\bM_i(k_0))=r_i$, which implies that all subsequent eigenvalues are theoretically equal to zero, that is, $\lambda_{l}^{(i)}(k_0)=0$ for any $l=r_i+1,r_i+2,\ldots,p$. Thus, equation (\ref{Qt@M@Q_k0}) admits an alternative spectral representation: 
\begin{equation}\label{Qt@M@Q_k0_simple}
   \bQ^{(i)}(k_0)^T\bM_{i}(k_0)\bQ^{(i)}(k_0)
   = \text{diag}(\lambda_{r_i+1}^{(i)}(k_0),\lambda_{r_i+2}^{(i)}(k_0),...,\lambda_{p}^{(i)}(k_0))=\mathbf{O}_{v_i}.
\end{equation}

Motivated by the projection of a second-order cross-moment matrix onto its orthogonal noise subspace in \cite{LiuZhang_2022}, we define the following criterion for estimating the change point:
\begin{equation}\label{G_score}
    G(k)=\sum_{i=1}^{2}\frac{\|\bQ^{(i)}(k)^T\bM_{i}(k)\bQ^{(i)}(k)\|_2}{\|\bM_{i}(k)\|_2},
\end{equation}
where the criterion function $G(k) \ge 0$ for all admissible $k$, and equality holds if and only if the candidate point matches the true change point $k=k_0$, as implied by (\ref{Qt@M@Q_k0}). Since $\bQ^{(i)}(k)$ is constructed from the eigendecomposition of $\bM_i(k)$ under the candidate partition, the criterion in (\ref{G_score}) is equivalently a candidate-specific residual eigenvalue criterion. It measures the deviation of $\bM_i(k)$ from the ideal rank-$r_i$ canonical correlation structure through the magnitude of its noise-space eigenvalues. The normalization by $\|\bM_{i}(k)\|_2$ is introduced to reduce the impact of regime-dependent scale variation across candidate partitions. In this way, the criterion function focuses on the relative magnitude of the residual component in the estimated noise subspace, rather than on the overall size of the canonical correlation matrix itself. 

By spectral decomposition, the proposed criterion function admits a simpler representation in terms of eigenvalue ratios:
\begin{equation}\label{G_score_simple}
    G(k)=\sum_{i=1}^{2}\frac{\lambda_{r_i+1}^{(i)}(k)}{\lambda_{1}^{(i)}(k)}.
\end{equation}
By (\ref{Qt@M@Q_k0_simple}), it follows immediately that $G(k_0)=0$. The condition $G(k_0)=0$ characterizes the ideal low-rank canonical correlation structure under the true partition. To see the identification mechanism more clearly, consider a candidate split $k \neq k_0$. If $k < k_0$, the second segment contains observations generated from both regimes; if $k > k_0$, the first segment is similarly contaminated by observations from the two regimes. In either case, the regime-specific CCA matrix computed under the candidate partition is no longer generated by a single homogeneous transformed factor structure. Provided that the factor loading spaces, or the corresponding dynamic canonical correlation structures, differ sufficiently across the two regimes, such a mixed segment induces additional nonzero canonical-correlation components outside the regime-specific signal space. Consequently, the residual eigenvalues in the candidate-specific noise subspace become positive, leading to $G(k)>G(k_0)$. This explains why the proposed criterion identifies the true change point at the population level.

Furthermore, replacing the matrix $L_2$ norm in (\ref{G_score}) by the Frobenius norm yields,
\begin{equation}\label{G_Fnorm}
    G_F(k)=\sum_{i=1}^{2}\frac{\|\bQ^{(i)}(k)^T\bM_{i}(k)\bQ^{(i)}(k)\|_F}{\|\bM_{i}(k)\|_F}.
\end{equation}
By definition of the Frobenius norm, (\ref{G_Fnorm}) can be rewritten as the ratio of cumulative squared noise eigenvalues to cumulative squared total eigenvalues, as shown in (\ref{G_Fnorm_lambda}),
\begin{equation}\label{G_Fnorm_lambda}
    G_{F}(k)=\sum_{i=1}^{2}\sqrt{\frac{\sum_{j=1}^{v_i}[\lambda_{r_i+j}^{(i)}(k)]^2}{\sum_{j=1}^{p}[\lambda_{j}^{(i)}(k)]^2}}.
\end{equation}
The corresponding change-point estimator under the Frobenius norm is defined analogously. The two criterion functions emphasize different aspects of the residual structure in the noise subspace. The $L_2$-norm-based criterion is more sensitive to the largest deviation from the ideal low-rank factor structure, whereas the Frobenius-norm-based criterion captures the cumulative contribution of all residual noise-space eigen-components. For this reason, both versions are retained and compared in the subsequent simulation study.

In empirical implementation, we compute the following sample covariance matrices before and after the candidate change point $k$,
\[
    \hat\bSigma_{\by,1}(k)=\frac{1}{k}\sum_{t=1}^{k}(\by_{t}-\bar\by^{(1)})(\by_{t}-\bar\by^{(1)})^T,\quad\text{and}\quad
    \hat\bSigma_{\by,2}(k)=\frac{1}{T-k}\sum_{t=k+1}^{T}(\by_{t}-\bar\by^{(2)})(\by_{t}-\bar\by^{(2)})^T,
\]
where $\bar\by^{(1)}=\dfrac{1}{k}\sum_{t=1}^{k}\by_{t}$ and $\bar\by^{(2)}=\dfrac{1}{T-k}\sum_{t=k+1}^{T}\by_{t}$. Analogous definitions are used for other sample covariance matrices. It follows that
\begin{equation}\label{M_hat}
    \hat\bM_i(k)=[\hat\bSigma_{\by,i}(k)]^{-1/2}\hat\bSigma_{\by\by_m,i}(k)[\hat\bSigma_{\by_m,i}(k)]^{-1}\hat\bSigma_{\by_m\by,i}(k)[\hat\bSigma_{\by,i}(k)]^{-1/2},\quad i=1,2.
\end{equation}
For each $i=1,2$ and $j=1,2,\ldots,p$, let $\hat\lambda_j^{(i)}(k)$ denote the estimate of the $j$-th largest eigenvalue of $\hat\bM_{i}(k)$, and let $\wh\bb_{j}^{(i)}(k)$ be the associated unit eigenvector. Then, if the number of factors $r_i$ is known, the estimator of the factor loading matrix in regime $i$ is given by $\wh\bL_1^{(i)}(k)=(\wh\bb_{1}^{(i)}(k),\wh\bb_{2}^{(i)}(k),...,\wh\bb_{r_i}^{(i)}(k))$. Finally, the factor vector can be estimated by (\ref{trans_factorepsilon}) as follows: 
\begin{equation}\label{factor_hat}
    \wh\bff_{t}^{(i)}(k)=[\wh\bL_1^{(i)}(k)]^T[\wh\bSigma_{\by,i}(k)]^{-1/2}\by_t \mathbb{I}_{t,i}(k), \quad i=1,2.
\end{equation} 
Next, we need to estimate the change point. Based on the preceding discussion, the sample criterion functions under the $L_2$ norm and the Frobenius norm are given by (\ref{G_simple_hat}) and (\ref{G_Fnorm_lambda_hat}), respectively:
\begin{equation}\label{G_simple_hat}
    \hat{G}(k)=\sum_{i=1}^{2}\frac{\wh\lambda_{r_i+1}^{(i)}(k)}{\wh\lambda_{1}^{(i)}(k)}.
\end{equation}
\begin{equation}\label{G_Fnorm_lambda_hat}
    \hat{G}_{F}(k)=\sum_{i=1}^{2}\sqrt{\frac{\sum_{j=1}^{v_i}[\wh\lambda_{r_i+j}^{(i)}(k)]^2}{\sum_{j=1}^{p}[\wh\lambda_{j}^{(i)}(k)]^2}}.
\end{equation}
Since $G(k_0)=0$ characterizes the true change point, we estimate its location by minimizing the empirical criterion function over the admissible search domain:
\begin{equation}\label{minG_hat}
    \hat{k} = \arg\min_{\left \lfloor \gamma_{1}T \right \rfloor \le k \le \left \lfloor \gamma_{2}T \right \rfloor}\hat G(k).
\end{equation}
Similarly, the change-point estimator under the Frobenius norm is defined by:
\begin{equation}\label{minG_F_hat}
    \hat{k} = \arg\min_{\left \lfloor \gamma_{1}T \right \rfloor \le k \le \left \lfloor \gamma_{2}T \right \rfloor}\hat G_F(k).
\end{equation}

A key identifying condition in the proposed procedure is that the idiosyncratic components are serially uncorrelated. This assumption ensures that, at the true change point $k_0$, the residual canonical correlation matrix in the noise subspace has exact zero eigenvalues, which forms the basis of the criterion function $G(k)$. This assumption is stronger than the weak-dependence conditions often imposed in approximate factor models, where idiosyncratic components may exhibit limited temporal or cross-sectional dependence. Its role here is to deliver a sharp zero-eigenvalue structure in the population criterion, which in turn yields a tractable identification argument for the change point.

\subsection{Estimation of the Number of Factors}
Estimating the regime-specific factor numbers $r_1$ and $r_2$ is an important component of the proposed procedure. We next present the one-step estimation method. For $i=1,2$, we can estimate the number of zero canonical correlations $v_i$ (and hence $r_i$) by testing the null hypothesis $H_0:\lambda_{p-v_{i}+1}=...=\lambda_p=0$ and $\lambda_{p-v_i}\neq0$ versus the alternative hypothesis $H_a:\lambda_{p-v_i}=0$. Let $n_1=\left \lfloor \gamma_{1}T \right \rfloor$ and $n_2=T-\left \lfloor \gamma_{2}T \right \rfloor$, which represent the sample sizes of the two boundary regions that exclude the possible change-point neighborhood, respectively. For simplicity, let $\wh\lambda_{i,j}(\gamma_i)$ be the $j$-th largest eigenvalue of $\wh\bM_i(\gamma_i)$, where $\wh\bM_i(\gamma_i)$ is the canonical correlation matrix computed from the boundary interval $[1,\left \lfloor\gamma_1T\right \rfloor]$ or $[\left \lfloor\gamma_2T\right \rfloor,T]$, respectively. A test statistic for testing the hypothesis of zero canonical correlations is given by 
\begin{equation}\label{HT_statistic}
    S_{n_i}(v)=-(n_i-m+1)\sum_{j=1}^{v}\log(1-\wh\lambda_{i,p-j+1}(\gamma_{i})),
\end{equation}
where the statistic $S_{n_i}(v_i)$ asymptotically follows the chi-square distribution with degrees of freedom $df_i=v_i[(m-1)p+v_i]$ under the null hypothesis. The above procedure conducts the test sequentially from $v_i=1$ until the null hypothesis is accepted, leading to the estimator
\begin{equation}\label{HT_method}
    \hat{r}_{i,HT}=p-\min_{1 \le v \le p-1}\{v:S_{n_i}(v) \le \chi^2_{\alpha}(df_i)\},
\end{equation}
where $\chi^2_{\alpha}(df_i)$ is the $\alpha$-quantile of the chi-square distribution.

Furthermore, we consider the transformed contribution ratio (TCR) estimator of \cite{Xia_TCR_2017} for comparison.


\begin{equation}\label{TCR}
    \hat{r}_{i,TCR}=\arg\min_{1 \le j \le r_{\max}}\frac{\log(1+\wh\lambda_{i,j+1}(\gamma_{i})/V_{i,j})}{\log(1+\wh\lambda_{i,j}(\gamma_{i})/V_{i,j-1})},
\end{equation}
where $V_{i,j}=\sum_{s=j+1}^{p}\wh\lambda_{i,s}(\gamma_i)$.

\begin{algorithm}[H]
\caption{Alternating Iterative Estimation (AIE)}
\label{algorithm1}
\KwIn{Data matrix \(Y\), convergence thresholds \(\varepsilon_1, \varepsilon_2\), maximum number of iterations \(K_{\max}\)}
\KwOut{Estimated change-point \(\hat{k}^*\) and regime-specific factor numbers \(\hat{r}_1^*\), \(\hat{r}_2^*\)}

Compute initial factor-number estimates from boundary data:
\[
(\hat{r}_1^{(0)}, \hat{r}_2^{(0)}) 
\gets 
\mathcal F(\lfloor \gamma_1 T \rfloor, \lfloor \gamma_2 T \rfloor).
\]

Compute the initial change-point estimate:
\[
\hat{k}^{(0)} 
\gets 
\arg \min_{\lfloor \gamma_1 T \rfloor \le k \le \lfloor \gamma_2 T \rfloor} 
\hat G(k; \hat{r}_1^{(0)}, \hat{r}_2^{(0)}).
\]

Set \(t \gets 0\).

\While{\(t < K_{\max}\)}{

Update factor-number estimates based on the current change-point estimate:
\[
(\hat{r}_1^{(t+1)}, \hat{r}_2^{(t+1)}) 
\gets 
\mathcal F(\hat{k}^{(t)}, \hat{k}^{(t)}+1).
\]

Update the change-point estimate based on the updated factor numbers:
\[
\hat{k}^{(t+1)} 
\gets 
\arg \min_{\lfloor \gamma_1 T \rfloor \le k \le \lfloor \gamma_2 T \rfloor} 
\hat G(k; \hat{r}_1^{(t+1)}, \hat{r}_2^{(t+1)}).
\]

Set \(t \gets t+1\).

\If{
\[
|\hat{k}^{(t)}-\hat{k}^{(t-1)}| \le \varepsilon_1
\]
\[
\text{and}\quad
\frac{1}{2}
\left(
|\hat{r}_1^{(t)}-\hat{r}_1^{(t-1)}|
+
|\hat{r}_2^{(t)}-\hat{r}_2^{(t-1)}|
\right)
\le \varepsilon_2
\]
}{
\textbf{break}
}
}

Set
\[
\hat{r}_1^* \gets \hat{r}_1^{(t)},\qquad
\hat{r}_2^* \gets \hat{r}_2^{(t)},\qquad
\hat{k}^* \gets \hat{k}^{(t)}.
\]

\end{algorithm}

The use of boundary data serves as a stable starting point for factor-number estimation. Since the two boundary segments are far away from the central search region for the change point, they are less likely to be contaminated by observations near the structural change point and therefore more likely to reflect relatively homogeneous factor regimes. This feature makes them suitable for constructing initial estimators of the factor dimension. Nevertheless, relying only on boundary data does not fully exploit the information contained in the entire sample. Therefore, the one-step method for estimating the number of factors described above has certain limitations. Conversely, using a larger portion of the sample requires reasonably accurate knowledge of the change-point location, creating an intrinsic circular dependency between change-point detection and factor-number estimation. To address this issue, 
we propose the Alternating Iterative Estimation (AIE) algorithm, as shown in Algorithm~\ref{algorithm1}, where \(\mathcal F(a,b)\) denotes a generic factor-number estimation rule based on the two subsamples \([1,a]\) and \([b,T]\). The AIE algorithm alternates between change-point estimation and factor-number selection until the two estimates stabilize. We first obtain initial estimates of the factor numbers from the boundary data and then estimate the change point based on these preliminary values. Given the updated change-point estimate, the sample is repartitioned and the factor numbers in the two regimes are re-estimated. Repeating these two steps allows the procedure to gradually reduce the circular dependence between the unknown change-point location and the unknown factor numbers. 

In this AIE algorithm, the initial estimators $\hat{r}^{(0)}_1$ and $\hat{r}^{(0)}_2$ are obtained by the same method (e.g., TCR or HT) using the boundary data, and the initial change-point estimate $\hat{k}^{(0)}$ is then computed by treating $\hat{r}^{(0)}_1$ and $\hat{r}^{(0)}_2$  as the number of factors. The iteration terminates when the successive changes in the factor-number estimates and the change-point estimate fall below the thresholds $\epsilon_2$ and $\epsilon_1$, respectively, or when the number of iterations reaches $K_{\max}$. Upon termination, the resulting estimators are reported.

\section{Theoretical Properties}\label{sec3}
We next establish the theoretical properties of the proposed estimators. The analysis is conducted under an $\alpha$-mixing condition on the joint process $\{(\by_t,\bff_t^{(1)},\bff_t^{(2)}):1\leq t\leq T\}$. The corresponding mixing coefficients are defined as follows:
\begin{equation}\label{amix}
\alpha_p(k)=\sup_{i}\sup_{A\in\mathcal{F}_{-\infty}^i,B\in \mathcal{F}_{i+k}^\infty}|P(A\cap B)-P(A)P(B)|,
\end{equation}
where $\mathcal{F}_i^j$ is the $\sigma$-field generated by $\{(\by_t,\bff_t^{(1)},\bff_t^{(2)}):i\leq t\leq j\}$.

\begin{assumption}\label{a1}
The joint process $\{(\by_t,\bff_t^{(1)},\bff_t^{(2)}):1\leq t\leq T\}$ is $\alpha$-mixing with mixing coefficients satisfying $\sum_{k=1}^\infty\alpha_p(k)^{1-2/\xi}=\alpha<\infty$ for some $\xi>2$, where $\alpha_p(k)$ is defined in (\ref{amix}).
\end{assumption}

\begin{assumption}\label{a2}
For any $i=1,2$, $j=1,...,p$, and $1 \le t \le T$, the observables satisfy the moment conditions $\mathbb{E}(y_{jt})=0$ and $\mathbb{E}|y_{jt}|^{2\xi}\leq c$, alongside the covariance matrix $\Cov(\bfeta_t^{(i)})=\bI_p$, where $c>0$ is a constant and $\xi$ is given in Assumption~\ref{a1}.
\end{assumption}

\begin{assumption}\label{a3} 
    The factor processes $\bff_t^{(1)}$ and $\bff_t^{(2)}$ before and after the change point $k_0$ are non-degenerate and stationary VAR($d$) processes.
\end{assumption}

Assumption~\ref{a1} is a standard regularity condition for weakly dependent time series, and its validity for VAR-type models is discussed in \cite{gaoetal2017}. For each $i=1,2$, the moment conditions $\mathbb{E}(y_{jt})=0$ and $\mathbb{E}|y_{jt}|^{2\xi}\leq c$ in Assumption~\ref{a2} can be justified under suitable conditions on the series $\bfeta_t^{(i)}$ and the corresponding loading matrix $\widetilde\bL^{(i)}$ in model (\ref{cpf-transform-simple}). For instance, these conditions hold if $\mathbb{E}(\eta_{jt}^{(i)})=0$, $\mathbb{E}|\eta_{jt}^{(i)}|^{2\xi}< \infty$, and $\|\widetilde\bL^{(i)}\|_{\infty}<\infty$, where $\|\cdot\|_{\infty}$ denotes the row norm of a matrix. 

\begin{assumption}\label{a4}
For each $i=1,2$, there exist positive constants $c_1$, $c_2$, $c_3$, $c_4$, $\kappa_1$ and $\kappa_2$ such that 
\[ c_1 \leq \sigma_{\min}^{+}\left(\bSigma_{\bff{\bfeta_m,i}}(k_0)\right)\leq \|\bSigma_{\bff{\bfeta_m,i}}(k_0)\|_{2}\leq  c_2, \] and 
\[c_3 \leq \lambda_{\min}(\bSigma_{\bfeta_m,i}(k_0))\leq \|\bSigma_{\bfeta_m,i}(k_0)\|_{2}\leq c_4.\] 
Furthermore, the loading matrix $\widetilde \bL^{(i)}$ and its factor-loading block $\widetilde \bL_1^{(i)}$
\[\kappa_1\leq \lambda_{\min}(\wt\bL^{(i)}\wt\bL^{(i)^T})\leq\lambda_{\max}(\wt\bL^{(i)}\wt\bL^{(i)^T})\leq \kappa_2,
\]
\[
\sigma_{\min}^{+}(\widetilde \bL_1^{(i)}) \geq \kappa_1^{1/2},
\qquad
\|\widetilde \bL_1^{(i)}\|_2 \leq \kappa_2^{1/2},
\]
where $\kappa_1$ and $\kappa_2$ are permitted to diverge with the dimension $p$. 
\end{assumption}

\begin{assumption}\label{a5}
The matrix $\bM_i(k_0)$ in each regime $i\in\{1,2\}$ admits $r_i$ distinct positive eigenvalues, such that $\lambda_1^{(i)}(k_0)>\cdots>\lambda_{r_i}^{(i)}(k_0)>0$.
\end{assumption}

\begin{assumption}\label{a6} 
There exists a lower threshold $\delta_0$ ensuring that $D\left(\mathcal{M}({\bf
L}_1^{(1)}),\mathcal{M}({{\bf
L}_1^{(2)}})\right) > \delta_0$ as both $p$ and $T$ go to infinity, where $D(\cdot,\cdot)$ denotes the discrepancy measure between the factor loading space $\mathcal{M}({\bf
L}_1^{(1)})$ before $k_0$ and the factor loading space $\mathcal{M}({\bf
L}_1^{(2)})$ after $k_0$. The measure is defined as: 
\begin{equation}\label{D_r1r2}
    D\left(\mathcal{M}({\bf
L}_1^{(1)}),\mathcal{M}({{\bf
L}_1^{(2)}})\right) = \sqrt{1-\frac{1}{\min(r_1,r_2)}\textrm{tr}({\bf
L}^{(1)}_1{\bf L}_1^{(1)^\T}{\bf L}^{(2)}_1{\bf L}_1^{(2)^\T})},
\end{equation}
where $\rank(\bL_1^{(1)})=r_1$ and  $\rank(\bL_1^{(2)})=r_2$. The measure (\ref{D_r1r2}) was first introduced in \cite{LiuChen_D_2020}. 
\end{assumption}

Under the normalization condition $\Cov(\bfeta_t^{(i)})=\bI_p$ in Assumption~\ref{a2} and the stationarity condition in Assumption~\ref{a3}, it follows as a mathematical corollary that the eigenvalues of $\bSigma_{\bff{\bfeta_m,i}}(k_0)$ and $\bSigma_{\bfeta_m,i}(k_0)$ are bounded in Assumption~\ref{a4}. The constants $\kappa_1$ and $\kappa_2$ of Assumption~\ref{a4} control the strength of the matrix $\wt\bL^{(1)}$ and $\wt\bL^{(2)}$, respectively. Assumption~\ref{a5} guarantees the identifiability of the eigenspaces associated with the positive eigenvalues of $\bM_i(k_0)$ for $i=1,2$. This allows Theorem 2 below to establish the convergence rates of the estimator for $\bL_1^{(i)}$ directly. Assumption~\ref{a6} requires the factor loading spaces before and after the change point to be sufficiently separated, thereby ensuring the identifiability of the structural change point.

\begin{theorem}
\label{tm1} Suppose that Assumptions 1-6 hold and that
$r_1$ and $r_2$ are known and fixed. Then
\[
\mathbb{P}\{\wh k \geq k_0+\epsilon T\}\leq~\left\{ 
\begin{array}{ll}
    CT^{-1/2}, & {\rm if\; \;}  p \text{\rm \;\; is fixed};\\[1ex]
    C(\epsilon^{-2}\kappa_1^{-2}\kappa_2+\epsilon^{-1}\kappa_1^{-4}\kappa_2^3)pT^{-1/2}, &{\rm if\;} p=o\{\min(T^{1/2},(\frac{\epsilon^{2}\kappa_1^{2}}{\kappa_2}+\frac{\epsilon\kappa_1^{4}}{\kappa_2^3})T^{1/2})\},
\end{array}
\right.
\] 
where 
$\frac{m+1}{T}<\epsilon<\gamma_2-\frac{k_0}{T}$, and 
$(\kappa_2\kappa_1^{-2}+\kappa_2^2\kappa_1^{-3})T^{-1}=o(\epsilon)$. A symmetric bound holds for the left-tail probability 
\(\mathbb P\{\widehat{k}<k_0-\epsilon T\}\) under the corresponding condition 
\((m+1)/T<\epsilon<k_0/T-\gamma_1\).

\end{theorem}
\begin{remark}\label{rm1}
Theorem~\ref{tm1} provides probability upper bounds for the deviation of the change-point estimator \(\widehat{k}\) from the true change point \(k_0\). The result shows explicitly how the estimation accuracy depends on the factor strength and the cross-sectional dimension through \(\kappa_1\), \(\kappa_2\), and \(p\). In particular, when \(\kappa_1 \asymp \kappa_2 \asymp \kappa\), Theorem~\ref{tm1} implies that
\[
\mathbb{P}\{\widehat{k} \geq k_0+\epsilon T\}
\leq C(1+\epsilon)\epsilon^{-2}\kappa^{-1}pT^{-1/2},
\]
provided that \(p=o\{\min(\sqrt{T},\epsilon(\epsilon+1)\kappa)\}\). Together with the corresponding left-tail bound, this yields
\[
|\widehat{k}-k_0|/T=o_p(1),
\]
which establishes the consistency of the proposed change-point estimator.
\end{remark}

\begin{theorem}
\label{tm2} Suppose that Assumptions 1-5 hold and that
$r_i$ is known and fixed. Then for $i=1,2$ and $j=1,2$,
\[
\|\widehat{{\bL}}_{j}^{(i)}(k_0)-{\bL}_{j}^{(i)}\|_2 =~\left\{
     \begin{array}{ll}
       O_p(
T^{-1/2}), & {\rm if\; \;}  p \text{\rm \;\; is fixed};
 \\[1ex]
       O_p(\kappa_1^{-2}\kappa_2pT^{-1/2}), \quad &{\rm if\;} p=o\{\min(T^{1/2},\kappa_2^{-1}\kappa_1^2T^{1/2})\}.
     \end{array}
   \right.
\]
\end{theorem}

Under a suitable orthogonal transformation of ${\bL}_{j}^{(i)}$, 
Theorem~\ref{tm2} establishes the convergence rate of the estimated loading
matrices at the true change point when the numbers of factors are known. For
each $i=1,2$, the choice of the semi-orthonormal loading matrix
$\bL_1^{(i)}$ in Model~\eqref{cpf-transform-simple2} is generally not unique. Consequently, we measure the estimation error in the column space $\mathcal{M}(\bL_1^{(i)})$ rather than in the matrix $\bL_1^{(i)}$ itself. This is because the column space $\mathcal{M}({{\bf L}}^{(i)}_1)$ is uniquely determined and remains invariant under different choices of $\bL_1^{(i)}$. We therefore define the estimated and true factor loading spaces as $\mathcal{M}(\wh{{\bf L}}^{(i)}_1)$ and $\mathcal{M}({{\bf L}}^{(i)}_1)$, respectively. The discrepancy measure proposed by \cite{panyao2008} is defined as 

\begin{equation}\label{dm}
\widetilde{D}\left(\mathcal{M}(\wh{{\bf L}}^{(i)}_1),\mathcal{M}({\bf
L}^{(i)}_1)\right)=\sqrt{1-\frac{1}{r_i}\textrm{tr}(\wh{{\bf
L}}^{(i)}_1\wh{{\bf L}}_1^{(i)^\T}{\bf L}^{(i)}_1{\bf L}_1^{(i)^\T})}.
\end{equation}
Note that $\widetilde{D}\left(\mathcal{M}(\wh{{\bf L}}^{(i)}_1),\mathcal{M}({\bf
L}^{(i)}_1)\right)$ takes values in $[0,1]$. In particular, it is equal to 0 if $\mathcal{M}(\wh{{\bf L}}^{(i)}_1)=\mathcal{M}({\bf L}^{(i)}_1)$ and is equal to 1 if
$\mathcal{M}(\wh{{\bf L}}^{(i)}_1)$ and $\mathcal{M}({\bf L}^{(i)}_1)$ are orthogonal.
Theorem~\ref{tm3} below establishes the convergence of $\widetilde{D}\left(\mathcal{M}(\wh{{\bf L}}^{(i)}_1),\mathcal{M}({\bf L}^{(i)}_1)\right)$ as $p$ and $T$ go to infinity.

\begin{theorem}
\label{tm3} Suppose Assumptions 1-5 hold. Assume further that
$r_i$ is known and fixed for each $i$. Then for any $i=1, 2$, 
\begin{equation} \label{ratesSP}
\widetilde{D}\left(\mathcal{M}(\widehat{\bf L}_{1}^{(i)}),\mathcal{M}({\bf L}_{1}^{(i)})\right) =~\left\{
     \begin{array}{ll}
       O_p(
T^{-1/2}), & {\rm if\; \;}  p \text{\rm \;\; is fixed}; \\[1ex]
       O_p(\kappa_1^{-2}\kappa_2pT^{-1/2}), \quad &{\rm if\;} p=o\{\min(T^{1/2},\kappa_2^{-1}\kappa_1^2T^{1/2})\}.
     \end{array}
   \right.
\end{equation}
\end{theorem}
Theorem~\ref{tm3} shows that, when $p$ is fixed, the estimator $\mathcal{M}(\widehat{\bf L}_{1}^{(i)})$ converges to the true loading space $\mathcal{M}({\bf L}_{1}^{(i)})$ at the standard rate of $T^{-1/2}$. Furthermore, under the additional condition  $\kappa_1\asymp\kappa_2\asymp p^{l}$ for some $0<l<1$ in Assumption~\ref{a4}, the convergence rate can be refined to $O_p(p^{1-l}T^{-1/2})$. This makes explicit how the interplay between factor strength and dimensionality influences loading-space recovery.

\section{Simulation Studies}\label{sec4}
This section investigates the finite-sample behavior of the proposed estimators through Monte Carlo experiments. The simulation design is chosen to reflect the theoretical features discussed in Section 3, with particular attention to the effects of sample size and cross-sectional dimension. We consider both the benchmark case in which the regime-specific factor numbers are known and the more realistic case in which they must be estimated from the data. The data-generating process follows Model (\ref{cpf-transform-simple}). We set the truncation parameters to $\gamma_1=0.1$ and $\gamma_2=0.9$. The true number of factors is set to $r_1=r_2=3$. We consider cross-sectional dimensions $p\in\{10,15,30,50\}$ and sample sizes $T\in\{200,500,1000\}$. For regimes $i=1,2$, the idiosyncratic components are generated from $\bve_{t}^{(i)}\sim \mathcal{N}(\bzero,\bI_{p-r_i})$. Within each regime, the latent factors $\bff_{t}^{(i)}$ are generated from a VAR(1) process of the form $\bff_{t}^{(i)} = \bPhi^{(i)} \bff_{t-1}^{(i)} + \bu_{t}^{(i)}$. The diagonal entries of the coefficient matrix are independently drawn from $U(0.2,0.9)$, and the innovation term $\bu_{t}^{(i)} \sim \mathcal{N}(\bzero,4\bI_{r_i})$. We focus on the setting in which the factor strengths remain unchanged across the change point. The entries of the loading matrices $\wt\bL^{(i)}$ are independently drawn from $U(-1,1)$. For each parameter configuration, we conduct 1,000 Monte Carlo replications to evaluate the performance of change-point estimation and loading-space recovery.

\begin{table}[htbp]
  \centering
  \caption{Average normalized absolute change-point estimation error, $|\widehat{k} - k_0|/T$, for the case where the factor numbers are known and satisfy $r_1 = r_2 = 3$.}
  \label{Table1}
  \begin{tabular}{c c c c c}
    \toprule
    && $T=200$ $(k_0=100)$ & $T=500$ $(k_0=250)$ & $T=1000$ $(k_0=500)$ \\
    \midrule
    $L_2$ norm & $p=10$ & 0.038 & 0.013 & 0.005 \\
    & $p=15$ & 0.040 & 0.011 & 0.003 \\
    & $p=30$ & 0.359 & 0.008 & 0.003 \\
    & $p=50$ & 0.387 & 0.085 & 0.003 \\
    \midrule
    $F$ norm & $p=10$ & 0.040 & 0.014 & 0.005 \\
    & $p=15$ & 0.039 & 0.013 & 0.005 \\
    & $p=30$ & 0.396 & 0.013 & 0.004 \\
    & $p=50$ & 0.395 & 0.063 & 0.007 \\
    \bottomrule
  \end{tabular}
\end{table}

\begin{table}[htbp]
  \centering
  \caption{Average loading-space discrepancy, $\dfrac{1}{2}\sum_{i=1}^{2}\bar{D}\left(\mathcal{M}(\widehat{\bL}^{(i)}_{1}),\mathcal{M}(\bL^{(i)}_{1})\right)$, for the case where the factor numbers are known and satisfy $r_1=r_2=3$.}
  \label{Table2}
  \begin{tabular}{l c c c c}
    \toprule
    && $T=200$ $(k_0=100)$ & $T=500$ $(k_0=250)$ & $T=1000$ $(k_0=500)$ \\
    \midrule
    $L_2$ norm & $p=10$ & 0.643 & 0.464 & 0.335 \\
    & $p=15$ & 0.764 & 0.573 & 0.426 \\
    & $p=30$ & 0.931 & 0.757 & 0.600 \\
    & $p=50$ & 0.975 & 0.879 & 0.732 \\
    \midrule
   $F$ norm & $p=10$ & 0.645 & 0.471 & 0.341 \\
    & $p=15$ & 0.766 & 0.586 & 0.441 \\
    & $p=30$ & 0.936 & 0.773 & 0.624 \\
    & $p=50$ & 0.975 & 0.897 & 0.764 \\
    \bottomrule
  \end{tabular}
\end{table}

\begin{table}[htbp]
  \centering
  \caption{Average normalized absolute change-point estimation error, $|\hat{k}-k_0|/T$, when the factor numbers $r_1$ and $r_2$ are unknown, with $p=10$.} 
  \label{Table3}
  \begin{threeparttable}
    \begin{tabular}{l l l c c c c c c}
      \toprule
      & & & \multicolumn{2}{c}{$T=200$} & \multicolumn{2}{c}{$T=500$} & \multicolumn{2}{c}{$T=1000$}\\
      \cmidrule(lr){4-5} \cmidrule(lr){6-7} \cmidrule(lr){8-9}
      & & & one-step & AIE & one-step & AIE & one-step & AIE \\
      \midrule
      based on $L_2$ norm & TCR & & 0.070 & 0.050 & 0.042 & 0.027 & 0.040 & 0.028 \\
      & HT & & 0.062 & 0.047 & 0.026 & 0.025 & 0.019 & 0.026 \\
      \midrule
      based on $F$ norm & TCR & & 0.043 & 0.039 & 0.028 & 0.015 & 0.014 & 0.007 \\
      & HT & & 0.085 & 0.052 & 0.036 & 0.020 & 0.012 & 0.012 \\
      \bottomrule
    \end{tabular}
    \begin{tablenotes}
      \footnotesize
      \item[1] ``one-step'' indicates that the factor numbers are estimated only once from the boundary data using the TCR or HT method, after which the change point is estimated accordingly.
      \item[2] ``AIE'' denotes the change-point estimator obtained from the alternating iterative estimation procedure, as shown in Algorithm~\ref{algorithm1}.
    \end{tablenotes}
  \end{threeparttable}
\end{table}

\begin{table}[htbp]
  \centering
  \caption{Average loading-space discrepancy, $\dfrac{1}{2}\sum_{i=1}^{2}\bar{D}(\mathcal{M}(\widehat{\bL}_1^{(i)}),\mathcal{M}(\bL_1^{(i)}))$, when the factor numbers $r_1$ and $r_2$ are unknown, with $p=10$.} 
  \label{Table4}
  \begin{tabular}{l l l c c c c c c}
    \toprule
    & & & \multicolumn{2}{c}{$T=200$} & \multicolumn{2}{c}{$T=500$} & \multicolumn{2}{c}{$T=1000$}\\
    \cmidrule(lr){4-5} \cmidrule(lr){6-7} \cmidrule(lr){8-9}
    & & & one-step & AIE & one-step & AIE & one-step & AIE \\
    \midrule
    based on $L_2$ norm & TCR & & 0.476 & 0.517 & 0.358 & 0.335 & 0.253 & 0.245 \\
    & HT & & 0.307 & 0.554 & 0.326 & 0.378 & 0.270 & 0.275 \\
    \midrule
    based on $F$ norm & TCR & & 0.454 & 0.516 & 0.346 & 0.342 & 0.235 & 0.244 \\
    & HT & & 0.309 & 0.561 & 0.331 & 0.382 & 0.270 & 0.268 \\
    \bottomrule
  \end{tabular}
\end{table}

To investigate the estimation performance of the proposed procedure, we first consider the benchmark case in which the factor numbers are known, namely $r_1=r_2=3$. Table~\ref{Table1} reports the average normalized absolute change-point estimation errors, $|\hat{k}-k_0|/T$. As illustrated in Table~\ref{Table1}, for a fixed dimension $p$, the average normalized error $|\hat{k}-k_0|/T$ decreases as the sample size $T$ increases. Conversely, for a fixed $T$, the estimation error tends to increase with the dimension $p$. The two norm-based criteria deliver broadly similar results, although their relative performance differs slightly across specific combinations of $p$ and $T$. For relatively small sample sizes ($T=200,500$), the estimation error exhibits an overall increasing trend as $p$ grows under both norm choices. This suggests that dimensionality has a more pronounced effect in smaller samples, whereas its impact becomes less substantial as the sample size increases. The two norm-based estimators yield broadly comparable performance, with only small differences manifesting under some combinations of $p$ and $T$. Overall, the simulation evidence supports the theoretical findings in Theorem~\ref{tm1} and indicates that the proposed CCA-based procedure becomes more accurate in larger samples.

We then examine the discrepancy between the true loading space and its estimate under the same settings for $T$ and $p$ as in the previous analysis. In this simulation study, we set ${\bf L}^{(i)}_1=\hat\bSigma_\by^{-\frac{1}{2}}\tilde{\bf L}^{(i)}_1$ for $i=1,2$, where $\hat\bSigma_\by$ denotes the sample covariance matrix of $\by_t$. Since $\tilde{\bf L}^{(i)}$ is not constructed to be strictly orthogonal in our simulation design, we modify the discrepancy measure (\ref{dm}) accordingly and use the criterion defined in (\ref{dmeasure}). Let $\bH_j$ be a $p\times r_j$ matrix with $\rank(\bH_j)=r_j$ for $j=1,2$. Define 
\begin{equation}\label{dmeasure}
\bar{D}(\mathcal{M}(\bH_1),\mathcal{M}(\bH_2))=\sqrt{1-
\frac{1}{\min{(r_1,r_2)}}\textrm{tr}(\bH_1(\bH_1^T\bH_1)^{-1} \bH_1^T\bH_2(\bH_2^T\bH_2)^{-1} \bH_2^T)}.
\end{equation}
This measure guarantees that $\bar{D}(\mathcal{M}(\bH_1),\mathcal{M}(\bH_2))=0$ if and only if one subspace is contained in the other (e.g. $\mathcal{M}(\bH_1)\subset \mathcal{M}(\bH_2)$), and $\bar{D}(\mathcal{M}(\bH_1),\mathcal{M}(\bH_2))=1$ if and only if $\mathcal{M}(\bH_1) \perp \mathcal{M}(\bH_2)$. Under the above definition, Table~\ref{Table2} reports the average discrepancy $\frac{1}{2}\sum_{i=1}^{2}\bar{D}(\mathcal{M}(\widehat{\bL}^{(i)}_{1}),\mathcal{M}(\bL^{(i)}_{1}))$ for the case where the factor numbers are known and equal to $r_1=r_2=3$. Table~\ref{Table2} shows that the subspace discrepancy decreases as $T$ increases for fixed $p$; in contrast, for fixed $T$, loading-space estimation becomes more difficult as the dimension $p$ increases. This behavior also agrees with the theoretical analysis, since loading-space recovery relies on accurate estimation of the leading eigenspaces of the canonical correlation matrix, a task that becomes more challenging in relatively high-dimensional settings. First, the subspace is a high-dimensional object, and finite-sample noise can produce substantial discrepancy even when the change point is accurately estimated. Second, the difficulty is exacerbated by the simulation design, particularly when the factor strengths or eigen-gaps are small, which naturally leads to larger deviations in the estimated subspace. Importantly, these discrepancies do not undermine the practical utility of the proposed method for change-point detection. The criterion function relies on minimizing residual canonical correlation in the noise subspace, which remains effective even if the estimated loading spaces deviate from the true ones.

We subsequently evaluate change-point estimation performance when $r_{1}$ and $r_{2}$ are treated as unknown and must be estimated from the data. We compare the one-step procedure with the proposed AIE algorithm for $T\in \{200, 500, 1000\}$, using both the Transformed Contribution Ratio (TCR) and Hypothesis Testing (HT) methods to estimate the unknown factor numbers $r_1$ and $r_2$. As detailed in Table~\ref{Table3}, for a fixed sample size and under the $L_2$-norm criterion, the HT approach generally yields smaller change-point estimation errors than the TCR approach. By contrast, under the Frobenius-norm criterion, the TCR method performs better than the HT method. Across most specifications, the AIE procedure improves upon the corresponding one-step estimator, suggesting that iterative updating helps mitigate the dependence between change-point estimation and factor-number estimation. The improvement is consistent with the purpose of the iterative scheme. By re-estimating the factor numbers after updating the break location, the AIE procedure uses information from more homogeneous regime-specific subsamples and reduces the dependence on the initial boundary-based estimates. The results suggest that the iterative update can be beneficial when the factor numbers are unknown. Taken together, the results in Table~\ref{Table3} suggest using HT with AIE under the $L_2$-norm criterion, and TCR with the AIE algorithm under the Frobenius-norm criterion.

\begin{figure}[htp]
\begin{center}   
\subfigure[]{\includegraphics[width=0.8\textwidth]{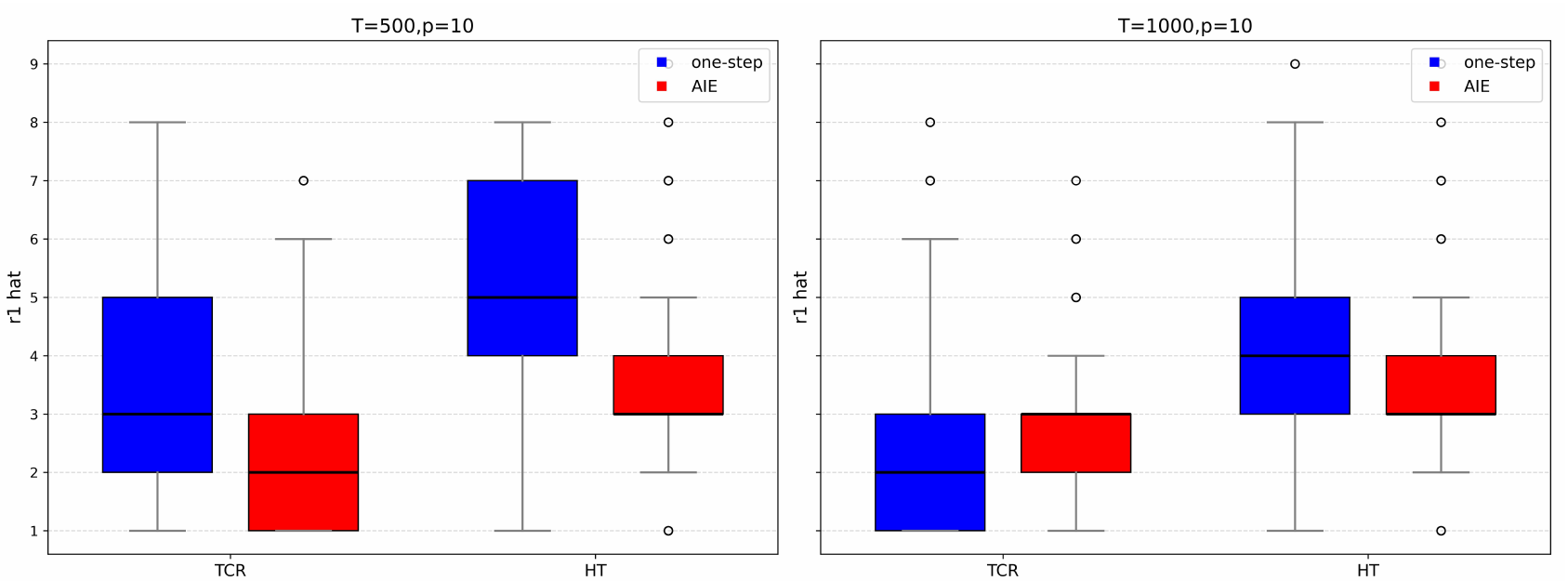}}\\ 
\subfigure[]{\includegraphics[width=0.8\textwidth]{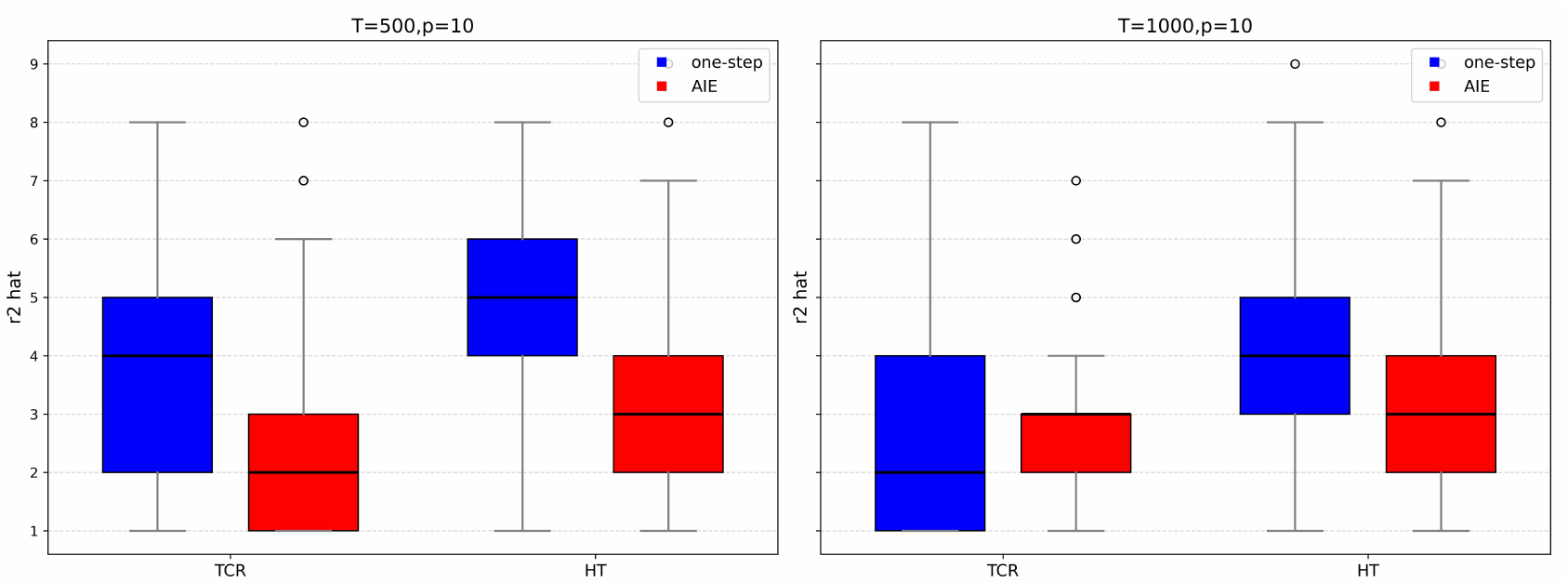}}
\end{center}
\caption{Boxplots of the estimated factor numbers under the one-step and alternating iterative estimation procedures: (a) $\hat{r}_1$; (b) $\hat{r}_2$. ``HT'' and ``TCR'' denote the hypothesis testing and transformed contribution ratio methods, respectively.}\label{r_hat}
\end{figure}

Figure~\ref{r_hat} presents box plots of the estimators $\hat{r}_{1}$ and $\hat{r}_{2}$ under different methods, constraining the settings to $T \in \{500, 1000\}$, $p=10$, and $r_1=r_2=3$. The plots indicate that the one-step estimators exhibit substantial variability and noticeable bias, although such dispersion gradually diminishes as $T$ increases. By contrast, the AIE algorithm produces factor-number estimates that are much more tightly concentrated around the true values, regardless of whether TCR or HT is used. In addition to reducing bias, the boxplots indicate that the AIE procedure substantially improves the stability of factor-number estimation, as reflected by the noticeably smaller dispersion of the resulting estimators.

Finally, Table~\ref{Table4} presents the discrepancy measure between the estimated and true factor loading spaces when the factor numbers are unknown. As $T$ increases, the discrepancy measure associated with the AIE algorithm remains comparable to that of the one-step estimator, indicating that the iterative estimation of the change point and factor numbers does not materially deteriorate loading-space estimation. Note that the loading-space discrepancies are relatively smaller than those
in Table~\ref{Table2}. {This phenomenon can be explained by the fact that the
estimated numbers of factors tend to exceed the true ones in this experiment.
Consequently, the column spaces of the estimated loading matrices are enlarged
and may provide a closer approximation to the true loading spaces, resulting
in smaller loading-space errors.}


\begin{figure}
\begin{center}
\centering
\includegraphics[width=0.7\textwidth]{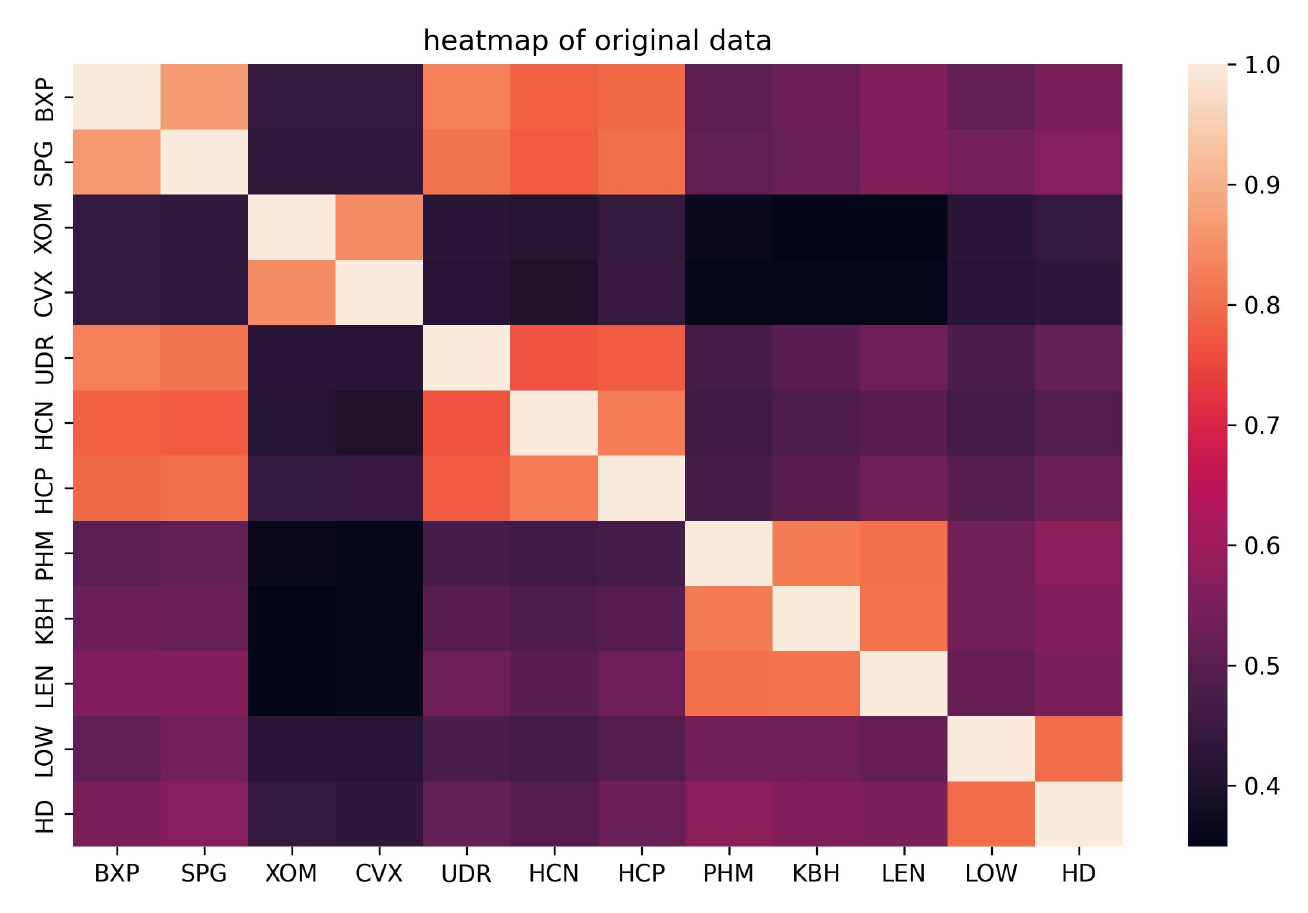}\\
\end{center}
\caption{The correlation heatmap of log returns for the 12 selected firms.}\label{heamap_12firms}
\end{figure}

\section{Empirical Analysis}\label{sec5}
This section applies the proposed method to two empirical datasets. The first application uses intraday log excess returns of S\&P 500 firms, and the second uses daily temperature records from U.S. monitoring stations. These applications are designed to illustrate how the proposed CCA-based procedure identifies interpretable changes in latent dynamic dependence structures in both financial and spatiotemporal panels.

{\noindent\bf Example 1 (Stock Returns).} The first application uses intraday log excess returns of S\&P 500 stocks from \cite{Pelger_2020}. The dataset covers the period from January 2004 to December 2016 and contains a balanced panel of 332 firms that are continuously observed over the full sample period. For illustration, we select 12 firms with relatively strong pairwise correlations and apply the proposed procedure to this subset. The correlation heatmap of these 12 firms is presented in Figure~\ref{heamap_12firms}. This subset is used mainly to provide a transparent visualization of the detected change point.

\begin{figure}[!htbp]
\begin{center}
\centering
\includegraphics[width=1\textwidth]{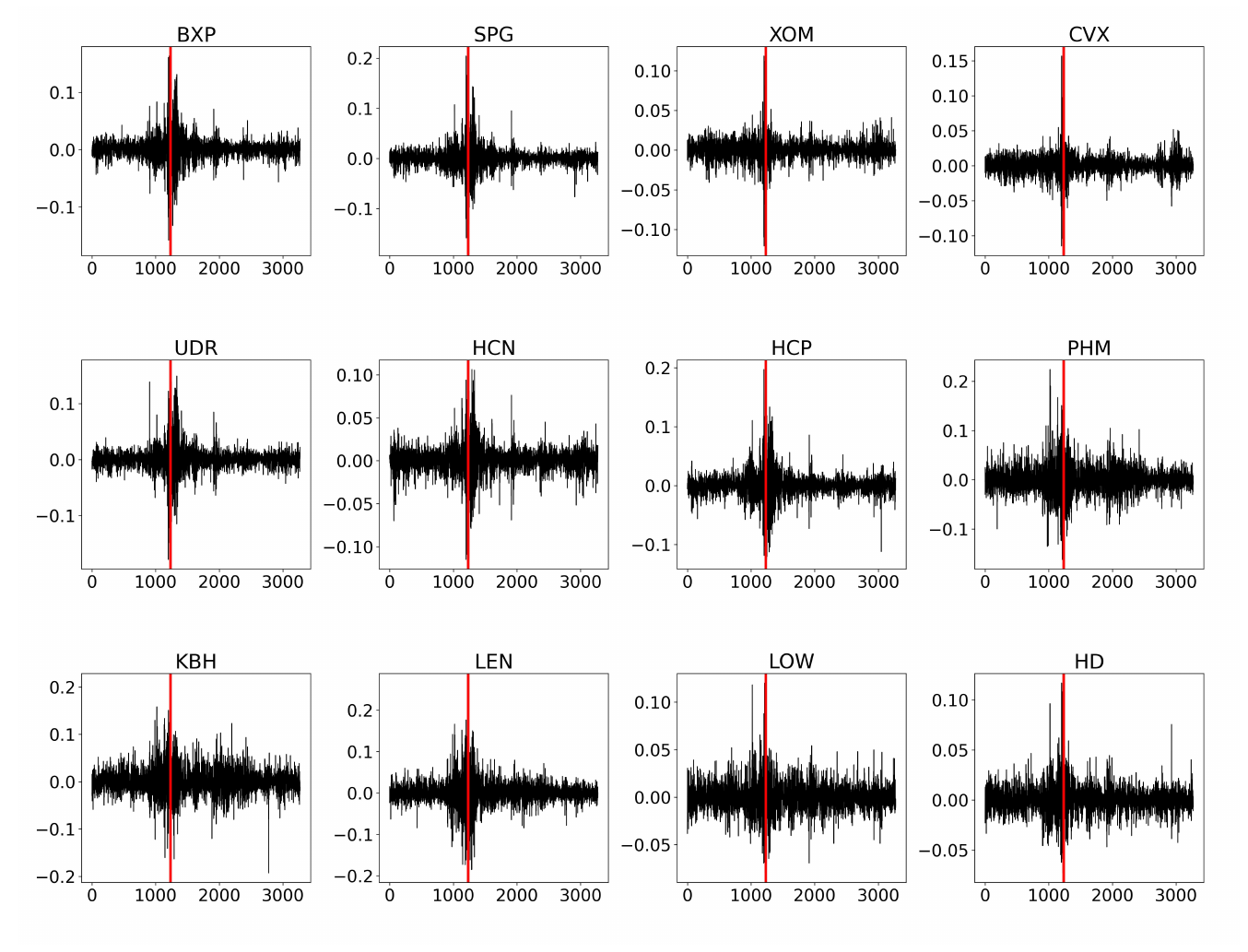}\\
\end{center}
\caption{Intraday return series of the 12 selected S\&P 500 firms, with the vertical red line marking the estimated change point on November 21, 2008.}\label{stock_return}
\end{figure}

For this dataset ($p=12,T=3213$), we set the maximum lag order to $m=2$ and restrict the change-point search interval to $[0.2, 0.8]$. Figure~\ref{stock_return} plots the intraday return series of the 12 selected S\&P 500 firms, with the vertical red line marking the estimated change point on November 21, 2008. Around the estimated change point, the return series exhibit a pronounced increase in volatility and more frequent extreme movements, particularly among real estate and housing-related firms such as BXP, SPG, UDR, HCP, PHM, KBH, and LEN, while the remaining stocks also display notable shifts in the magnitude of fluctuations. These synchronous changes across a broad set of firms indicate that the detected change point is unlikely to be driven by idiosyncratic shocks alone.

The estimated number of factors increases from $\widehat r_1=2$ before the change point to $\widehat r_2=4$ afterward. This increase suggests that the dimension of the latent dependence structure expanded during the crisis period. Economically, the estimated date lies within a period of severe financial stress in the United States, when concerns about major financial institutions, including Citigroup, intensified and shocks were transmitted from the financial sector to the broader economy. The result is therefore consistent with a reconfiguration of systematic risk following the Lehman Brothers collapse.

\begin{figure}[!htbp]
\begin{center}
\centering
\includegraphics[width=0.9\textwidth]{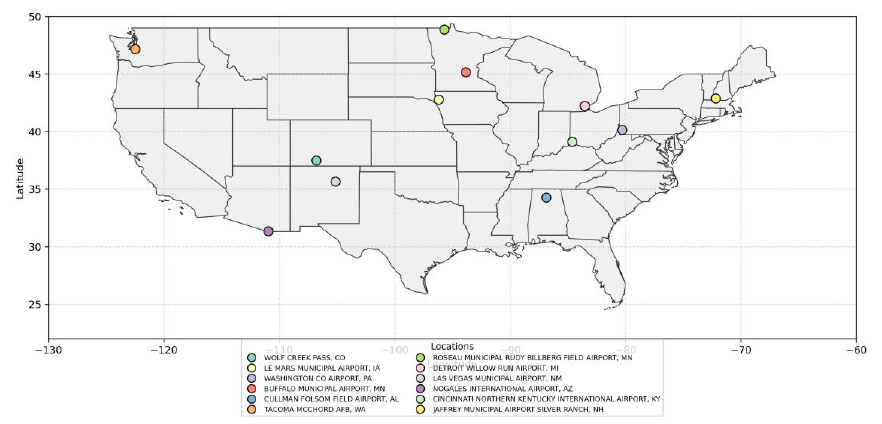}\\
\end{center}
\caption{Geographic distribution of the 12 meteorological stations in the contiguous United States.}\label{Ustemp_map}
\end{figure}

{\noindent\bf Example 2 (U.S. Temperatures).} 
We apply the proposed method to a 12-dimensional spatiotemporal dataset of daily surface temperatures. The data are obtained from the NOAA GSOD archive (2023); this dataset contains daily temperature records (in Fahrenheit) from January 1 to December 31, 2023, for 12 monitoring stations in the contiguous United States. Their geographic locations are shown in Figure~\ref{Ustemp_map}.

These monitoring stations are distributed across 11 states in the contiguous United States. More specifically, the stations cover a broad range of climatic regions across the Western, Midwestern, Eastern, and Southern United States. This broad geographic coverage provides reasonable spatial representation for the temperature dynamics under study. The raw temperature series for the individual stations are shown in Figure~\ref{Original series}. Notably, the series exhibit substantially greater variability from January to May than from June to August. This visual pattern suggests the possible presence of a change point during 2023.

\begin{figure}[h]
\begin{center}
\centering
\includegraphics[width=0.8\textwidth]{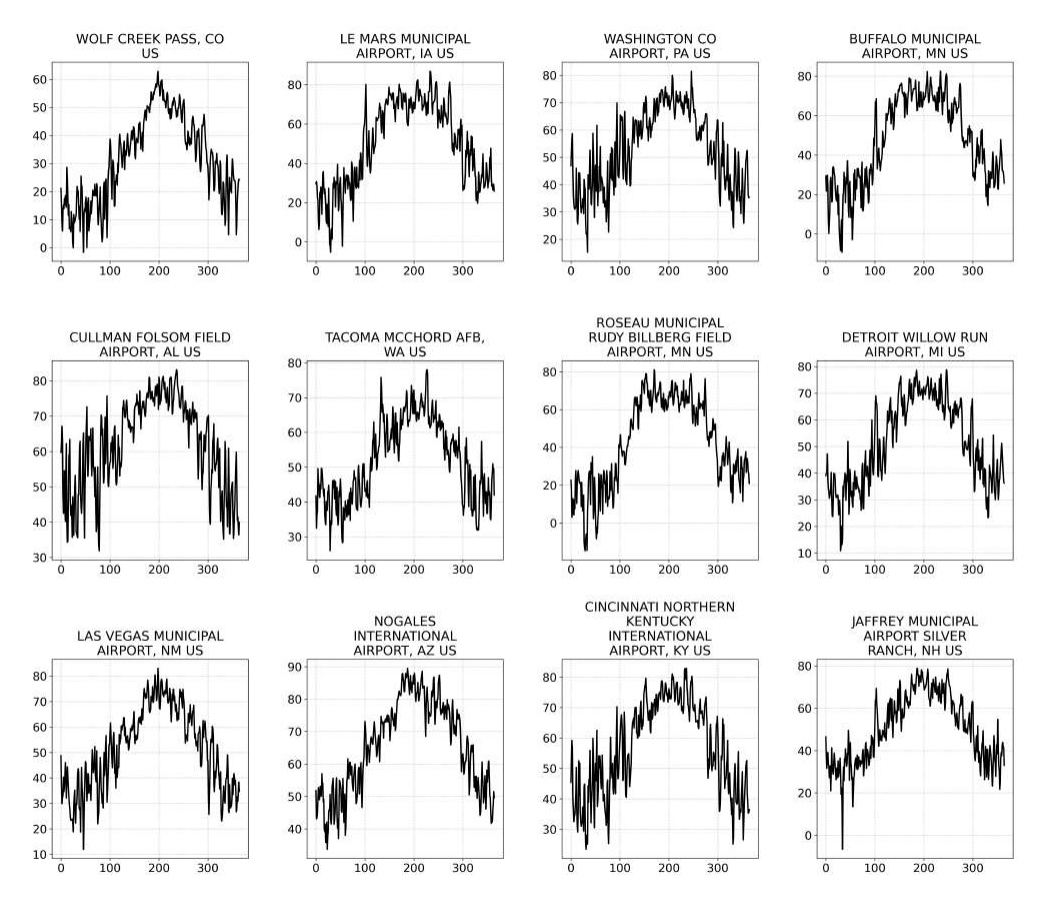}\\
\end{center}
\caption{Daily temperature series at the 12 stations from January 1, 2023, to December 31, 2023.}\label{Original series}
\end{figure}

After applying Seasonal-Trend decomposition using Loess (STL) to the original series, we set the maximum lag order to $m=2$ and restrict the change-point search interval to $[0.2, 0.8]$ in view of typical seasonal meteorological patterns. The estimated change point occurs on April 2, 2023. This estimate is consistent with the preliminary visual impression from the series and is marked in Figure~\ref{temp_cp_series}. It may be interpreted as reflecting a seasonal transition from the cold season to the warm season over mid-latitude North America. The estimation results indicate a reduction in factor numbers across the two regimes, with $\hat{r}_{1}=4$ factors in the more heterogeneous cold-season regime and $\hat{r}_{2}=2$ factors in the comparatively more homogeneous warm-season regime. The sign pattern of the first factor loadings changes little across the estimated change point, suggesting a broadly synchronous temperature anomaly pattern across locations, with the main difference lying in the loading magnitudes across the two regimes. The second post-change factor appears to capture a latitudinal contrast in temperature anomalies, roughly separated around the 42nd parallel north.

\begin{figure}
\begin{center}
\centering
\includegraphics[width=0.7\textwidth]{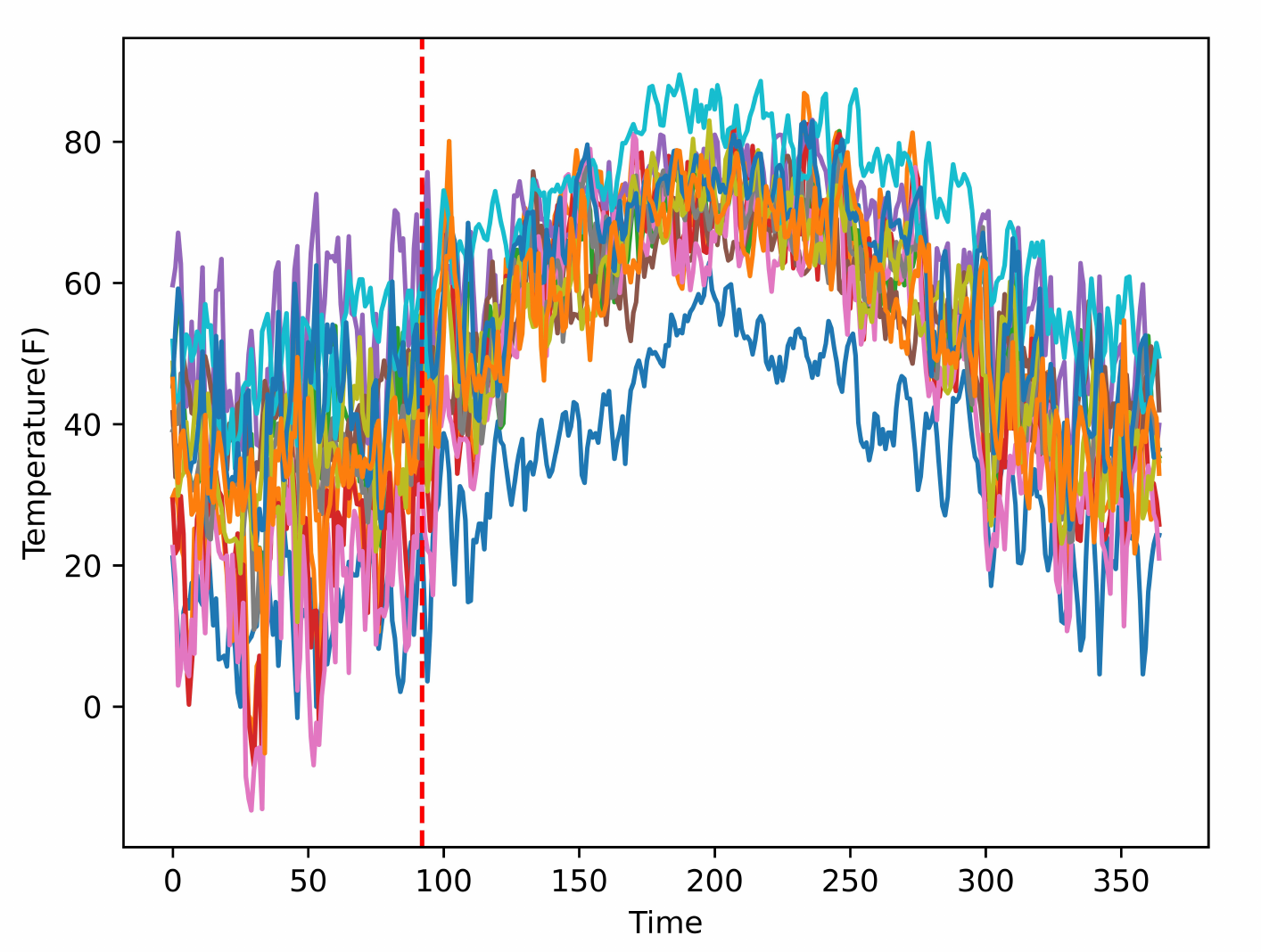}\\
\end{center}
\caption{Estimated change point in the original temperature series for the 12 monitoring stations.}\label{temp_cp_series}
\end{figure}

\section{Conclusion}\label{sec6}
This paper proposes a CCA-based approach to the estimation of a single change point in high-dimensional transformed factor models. The method uses the low-rank structure of regime-specific canonical correlation matrices to construct an eigenvalue-based criterion for identifying changes in the latent factor structure. When the regime-specific factor numbers are unknown, an alternating iterative estimation procedure is introduced to update the factor numbers and the change point sequentially.

The theoretical analysis establishes convergence properties of the proposed estimators under mixing and moment conditions. The results show explicitly how the estimation accuracy depends on factor strength, sample size, and cross-sectional dimension. Monte Carlo experiments provide numerical evidence consistent with these theoretical findings. The empirical applications to intraday stock returns and daily temperature series further illustrate that the method can identify interpretable changes in latent dependence structures.

Several extensions remain possible. One direction is to allow for multiple change points in transformed factor models. Another is to relax the serial uncorrelatedness assumption on the idiosyncratic component and study settings with weakly dependent idiosyncratic errors. These extensions would broaden the applicability of the proposed framework to more general high-dimensional time-series environments.

\section*{Data availability statement}
The datasets used in the real-data analysis of this paper are all publicly available from open-access repositories. Their sources are listed below for the two empirical examples. Example 1 uses the intraday stock returns dataset, which is publicly available at \url{https://mpelger.people.stanford.edu/data-and-code}.  Example 2 uses the global summary of the day dataset, which is accessible from the National Centers for Environmental Information (NCEI) archive at \url{https://www.ncei.noaa.gov/data/global-summary-of-the-day/archive/}. All datasets can be downloaded and used for research purposes subject to the respective access and usage policies of the hosting repositories. 

\section*{Author contributions}
{CRediT:} \textbf{Lei Jia}: Conceptualization, Methodology, Formal analysis, Investigation, Software, Validation, Data curation, Writing -- original draft, Visualization. \textbf{Shouri Hu}: Conceptualization, Methodology, Supervision, Writing -- review \& editing. \textbf{Zhaoxing Gao}: Conceptualization, Methodology, Supervision, Project administration, Writing -- review \& editing, Funding acquisition.

\section*{Disclosure statement}
{No potential conflict of interest was reported by the authors.}

\section*{Funding}
{This work was partially supported by the National Natural Science Foundation of China (NSFC) under Grant Nos. 72573029, U23A2064, 12201558, and 12401336.}
\section*{Appendix: Mathematical Proofs}
We denote $\mathcal{K}_T=\{\left \lfloor \gamma_{1}T \right \rfloor,\ldots,\left \lfloor \gamma_{2}T \right \rfloor\}$. We use $C$ to denote a generic positive constant whose value may change at different places. We begin with several useful lemmas.
\begin{lemma}\label{lemma_mkt}
For any $k\in \mathcal{K}_T$, we have $k_0 \asymp T$, $k-m \asymp T$, and $T-k \asymp T$, where $m\geq d$ is fixed. 
\end{lemma}
\textbf{\proofname.} Since the true change point $k_0 \in \mathcal{K}_T$ and $k \asymp T$, it follows that $k_0 \asymp T$. The remaining conclusions follow immediately when $m$ is fixed.

\begin{lemma}\label{sigma_bound}
If Assumption~\ref{a4} holds, then
\[
\kappa_1 \le \lambda_{\min}(\bSigma_{\by,i}(k_0))\le \|\bSigma_{\by,i}(k_0)\|_2\le\kappa_2,
\]
\[
\kappa_1c_3 \le \lambda_{\min}(\bSigma_{\by_m,i}(k_0))\le \|\bSigma_{\by_m,i}(k_0)\|_2\le\kappa_2c_4,
\]
and
\[
\kappa_1c_1 \leq \sigma_{\min}^{+}(\bSigma_{\by\by_m,i}(k_0)) \leq \|\bSigma_{\by\by_m,i}(k_0)\|_2\le\kappa_2c_2.
\]

\end{lemma}
\textbf{\proofname}. We first prove the bounds for $\bSigma_{\by,i}(k_0)$. For $i=1,2$, since $\Cov(\bfeta_t^{(i)})=\bI_p$,
\begin{center}
    $\bSigma_{\by,i}(k_0)=\wt\bL^{(i)}\Cov(\bfeta_t^{(i)})\wt\bL^{(i)^T}=\wt\bL^{(i)}\wt\bL^{(i)^T}$.
\end{center}
Hence, by Assumption~\ref{a4},
\[
\kappa_1 \le \lambda_{\min}(\bSigma_{\by,i}(k_0))\le \|\bSigma_{\by,i}(k_0)\|_2\le \kappa_2.
\]
Next, by the definition of $\by_{t,m}$,
\[
\bSigma_{\by_m,i}(k_0)=\Cov(\by_{t,m}\mathbb{I}_{t,i}(k_0))
=(\bI_m\otimes \wt\bL^{(i)})\bSigma_{\bfeta_m,i}(k_0)(\bI_m\otimes \wt\bL^{(i)})^T.
\]
Therefore,
\begin{align*}
\lambda_{\min}(\bSigma_{\by_m,i}(k_0))
&\ge \lambda_{\min}\{(\bI_m\otimes \wt\bL^{(i)})(\bI_m\otimes \wt\bL^{(i)})^T\}\,\lambda_{\min}(\bSigma_{\bfeta_m,i}(k_0))
\ge \kappa_1 c_3,
\end{align*}
and
\[
\|\bSigma_{\by_m,i}(k_0)\|_2
\le \|\bI_m\otimes \wt\bL^{(i)}\|_2^2\,\|\bSigma_{\bfeta_m,i}(k_0)\|_2
\le \kappa_2 c_4.
\]
Similarly, under the serial-uncorrelated idiosyncratic noise assumption,
\[
\bSigma_{\by\by_m,i}(k_0)=\Cov(\by_t\mathbb{I}_{t,i}(k_0),\by_{t,m}\mathbb{I}_{t,i}(k_0))
=\wt\bL_1^{(i)}\bSigma_{\bff\bfeta_m,i}(k_0)(\bI_m\otimes \wt\bL^{(i)})^T.
\]
Hence, we can obtain:
\begin{align*}
\sigma_{\min}^{+}(\bSigma_{\by\by_m,i}(k_0))
&\ge \sigma_{\min}^{+}(\wt\bL_1^{(i)})\sigma_{\min}^{+}(\bSigma_{\bff\bfeta_m,i}(k_0))\sigma_{\min}^{+}(\bI_m\otimes \wt\bL^{(i)})
\ge \kappa_1 c_1,
\end{align*}
and
\[
\|\bSigma_{\by\by_m,i}(k_0)\|_2
\le \|\wt\bL_1^{(i)}\|_2\,\|\bSigma_{\bff\bfeta_m,i}(k_0)\|_2\,\|\bI_m\otimes \wt\bL^{(i)}\|_2
\le \kappa_2 c_2.
\]
This completes the proof.


\begin{lemma}\label{lemma_sigma}
If Assumptions 1 and 2 hold, for any $k\ge k_0+\epsilon T$ and fixed $m$,
\[
\mathbb{E}\|\hat\bSigma_{\by,1}(k)-\bSigma_{\by,1}(k)\|_2 \le Cp(k-m)^{-1/2},\mathbb{E}\|\hat\bSigma_{\by,2}(k)-\bSigma_{\by,2}(k)\|_2 \le Cp(T-k)^{-1/2},
\]
\[
\mathbb{E}\|\hat\bSigma_{\by_m,2}(k)-\bSigma_{\by_m,2}(k)\|_2 \le Cp(T-k)^{-1/2},\mathbb{E}\|\hat\bSigma_{\by_m,1}(k)-\bSigma_{\by_m,1}(k)\|_2 \le Cp(k-m)^{-1/2},
\]
\[
\mathbb{E}\|\hat\bSigma_{\by\by_m,2}(k)-\bSigma_{\by\by_m,2}(k)\|_2 \le Cp(T-k)^{-1/2},\mathbb{E}\|\hat\bSigma_{\by\by_m,1}(k)-\bSigma_{\by\by_m,1}(k)\|_2 \le Cp(k-m)^{-1/2}.
\]

\end{lemma}
\textbf{\proofname.} We define $\bGamma_\by(a,b)=\Cov(\by_t \mathbb{I}_{\{a\leq t \leq b\}})$ as the corresponding cross moment matrix, and let $\wh \bGamma_\by (a,b)$ denote its sample counterpart computed over the interval $[a,b]\subset[1,T]$. Then
\begin{align*}
    \mathbb{E}\|\wh\bSigma_{\by,1}(k)-\bSigma_{\by,1}(k)\|_F^{2} &\leq (\frac{k_0-m}{k-m})^2\mathbb{E}\|\wh\bGamma_\by(m+1,k_0)-\bGamma_\by(m+1,k_0)\|_F^2\\
    &\quad+(\frac{k-k_0}{k-m})^2\mathbb{E}\|\wh\bGamma_\by(k_0+1,k)-\bGamma_\by(k_0+1,k)\|_F^2 \\
    &=(\frac{k_0-m}{k-m})^2 I_1+(\frac{k-k_0}{k-m})^2 I_2.
\end{align*}
By Assumption~\ref{a1}, Assumption~\ref{a2} and the inequality of \cite{Davydov_1968}, we have
\begin{align*}
    I_1&=\mathbb{E}\left\{\sum_{i=1}^p\sum_{j=1}^p|\frac{1}{k_0-m}\sum_{t=m+1}^{k_0}(y_{it}-\bar y_{i.})(y_{jt}-\bar y_{j.})-\mathbb{E}(y_{it}y_{jt})|^2\right\}\\
    &\leq p^2 \max_{1\le i,j\le p}\mathbb{E}|\frac{1}{k_0-m}\sum_{t=m+1}^{k_0}(y_{it}-\bar y_{i.})(y_{jt}-\bar y_{j.})-\mathbb{E}(y_{it}y_{jt})|^2\\
    &\leq p^2 \max_{1\le i,j\le p}\left\{2\mathbb{E}|\frac{1}{k_0-m}\sum_{t=m+1}^{k_0}(y_{it}y_{jt}-\mathbb{E}(y_{it}y_{jt}))|^2+2\mathbb{E}(\bar y_{i.}\bar y_{j.})^2\right\}\\
    &\leq p^2 \max_{1\le i,j\le p}\left\{ \frac{2}{(k_0-m)^2}\sum_{t=m+1}^{k_0}\mathbb{E}(y_{it}y_{jt}-\mathbb{E}(y_{it}y_{jt}))^2\right.\\
    &\quad +\left.\frac{2}{(k_0-m)^2}\sum_{t_1\ne t_2}\mathbb{E}\left(y_{it_1}y_{jt_1}-\mathbb{E}(y_{it_1}y_{jt_1})\right)(y_{it_2}y_{jt_2}-\mathbb{E}(y_{it_2}y_{jt_2}))+O\left((k_0-m)^{-2}\right)\right\}\\   
    &\leq p^2\left(\frac{2c^2}{(k_0-m)}+\frac{2}{(k_0-m)^2}\cdot\sum_{t_1\ne t_2}\alpha_p(|t_1-t_2|)^{1-2/\xi} + O\left((k_0-m)^{-2}\right)\right)\\
    &\leq \frac{2p^2(c^2+c_0\alpha)}{k_0-m} + O\left((k_0-m)^{-2}\right).
\end{align*}
Similarly, $I_2\le \dfrac{2p^2(c^2+c_0\alpha)}{k-k_0} + O\left((k-k_0)^{-2}\right)$. Then
\begin{align*}
    \mathbb{E}\|\wh\bSigma_{\by,1}(k)-\bSigma_{\by,1}(k)\|_2^{2}
    &\leq \mathbb{E}\|\wh\bSigma_{\by,1}(k)-\bSigma_{\by,1}(k)\|_F^{2}\\
    &\leq (\frac{k_0-m}{k-m})^2 I_1+(\frac{k-k_0}{k-m})^2 I_2\\
    &\leq 2p^2(c^2+c_0\alpha)(k-m)^{-1}+O\left((k-m)^{-2}\right).
\end{align*}
Hence, there exists a constant $C$ such that $\mathbb{E}\|\hat\bSigma_{\by,1}(k)-\bSigma_{\by,1}(k)\|_2 \le Cp(k-m)^{-1/2}$. The remaining bounds can be proved similarly.

\begin{lemma}\label{lemma_4}
Let Assumptions 1-4 hold. For $k \ge k_0 + \epsilon T$, then,
\begin{center}
     $(i)\quad \mathbb{E}\|\hat\bSigma_{\by,2}^{-1/2}(k)-\bSigma_{\by,2}^{-1/2}(k)\|_2 \le C\kappa_1^{-\frac{3}{2}}p(T-k)^{-1/2}$,\\
    $(ii)\quad \mathbb{E}\|\hat\bSigma_{\by_m,2}^{-1/2}(k)-\bSigma_{\by_m,2}^{-1/2}(k)\|_2 \le C\kappa_1^{-\frac{3}{2}}p(T-k)^{-1/2}$,
    $(iii)\quad \mathbb{E}\|\hat\bSigma_{\by,2}^{-1/2}(k)\hat\bSigma_{\by\by_m,2}(k)-\bSigma_{\by,2}^{-1/2}(k)\bSigma_{\by\by_m,2}(k)\|_2 \le C(\kappa_1^{-\frac{3}{2}}\kappa_2+\kappa_1^{-\frac{1}{2}})p(T-k)^{-1/2}$.   
\end{center}
\end{lemma}
\textbf{\proofname.} By Lemma~\ref{sigma_bound}, Lemma~\ref{lemma_sigma} and Lemma 2.2 of \cite{schmitt1992}, 
\begin{align*}
    \mathbb{E}\|\hat\bSigma_{\by,2}^{-1/2}(k)-\bSigma_{\by,2}^{-1/2}(k)\|_2 &=\mathbb{E}\|-\wh \bSigma_{\by,2}^{-1/2}(k)\left(\hat\bSigma_{\by,2}^{1/2}(k)-\bSigma_{\by,2}^{1/2}(k)\right)\bSigma_{\by,2}^{-1/2}(k)\|_2\\
    &\leq \kappa_1^{-1}\mathbb{E}\|\hat\bSigma_{\by,2}^{1/2}(k)-\bSigma_{\by,2}^{1/2}(k)\|_2\\
    &\leq 2^{-1/2}\kappa_1^{-3/2}\mathbb{E}\|\hat\bSigma_{\by,2}(k)-\bSigma_{\by,2}(k)\|_2\\
    &\leq C\kappa_1^{-3/2}p(T-k)^{-1/2}.
\end{align*}
The inequality (ii) is proved similarly, so we have
\begin{center}
    $\mathbb{E}\|\hat\bSigma_{\by_m,2}^{-1/2}(k)-\bSigma_{\by_m,2}^{-1/2}(k)\|_2 \le C\kappa_1^{-\frac{3}{2}}p(T-k)^{-1/2}$.
\end{center}
For the third inequality, note that
\begin{align*}
    \hat\bSigma_{\by,2}^{-1/2}(k)\hat\bSigma_{\by\by_m,2}(k)-\bSigma_{\by,2}^{-1/2}(k)\bSigma_{\by\by_m,2}(k)&=(\hat\bSigma_{\by,2}^{-1/2}(k)-\bSigma_{\by,2}^{-1/2}(k))\hat\bSigma_{\by\by_m,2}(k)\\
    &\quad+\bSigma_{\by,2}^{-1/2}(\hat\bSigma_{\by\by_m,2}(k)-\bSigma_{\by\by_m,2}(k)).
\end{align*}
Hence there exists a constant $C$ such that 
\[\mathbb{E}\|\hat\bSigma_{\by,2}^{-1/2}(k)\hat\bSigma_{\by\by_m,2}(k)-\bSigma_{\by,2}^{-1/2}(k)\bSigma_{\by\by_m,2}(k)\|_2 \le C(\kappa_1^{-\frac{3}{2}}\kappa_2+\kappa_1^{-\frac{1}{2}})p(T-k)^{-1/2}.
\] This completes the proof.

\begin{lemma}\label{lemma_4_addition}
Similar to Lemma~\ref{lemma_4}, we have the following result:
\begin{center}
     $(i)\quad \mathbb{E}\|\hat\bSigma_{\by,1}^{-1/2}(k)-\bSigma_{\by,1}^{-1/2}(k)\|_2 \le C\kappa_1^{-\frac{3}{2}}p(k-m)^{-1/2}$,\\
    $(ii)\quad \mathbb{E}\|\hat\bSigma_{\by_m,1}^{-1/2}(k)-\bSigma_{\by_m,1}^{-1/2}(k)\|_2 \le C\kappa_1^{-\frac{3}{2}}p(k-m)^{-1/2}$,
    $(iii)\quad \mathbb{E}\|\hat\bSigma_{\by,1}^{-1/2}(k)\hat\bSigma_{\by\by_m,1}(k)-\bSigma_{\by,1}^{-1/2}(k)\bSigma_{\by\by_m,1}(k)\|_2 \le C(\kappa_1^{-\frac{3}{2}}\kappa_2+\kappa_1^{-\frac{1}{2}})p(k-m)^{-1/2}$.   
\end{center}
\end{lemma}
\textbf{\proofname.} When $m+1<\epsilon T<\gamma_2T-k_0$, the mixed covariance matrices before $k$ can be written in the following weighted forms:

\[\bSigma_{\by,1}(k)=\frac{k_0-m}{k-m}\bSigma_{\by,1}(k_0)+\frac{k-k_0}{k-m}\bSigma_{\by,2}(k_0),\]
\[\bSigma_{\by\by_m,1}(k)=\frac{k_0-m}{k-m}\bSigma_{\by\by_m,1}(k_0)+\frac{k-k_0-m}{k-m}\bSigma_{\by\by_m,2}(k_0)+\frac{1}{k-m}\sum_{t=k_0+1}^{k_0+m}\mathbb E(\by_{t}\by_{t,m}^T),\]
\[\bSigma_{\by_m,1}(k)=\frac{k_0-m}{k-m}\bSigma_{\by_m,1}(k_0)+\frac{k-k_0-m}{k-m}\bSigma_{\by_m,2}(k_0)+\frac{1}{k-m}\sum_{t=k_0+1}^{k_0+m}\mathbb E(\by_{t,m}\by_{t,m}^T).\]

Let $\beta=\dfrac{k_0-m}{k-m}$ and $1-\beta=\dfrac{k-k_0}{k-m}=\dfrac{k-k_0-m}{k-m}+\dfrac{m}{k-m}$, then $\bSigma_{\by\by_m,1}(k)$ and $\bSigma_{\by_m,1}(k)$ admit the following representations involving the weight $\beta$:
\begin{center}
    $\bSigma_{\by\by_m,1}(k)=\beta\bSigma_{\by\by_m,1}(k_0)+(1-\beta)\bSigma_{\by\by_m,2}(k_0)+\bGamma_{1,m},$\\
    $\bSigma_{\by_m,1}(k)=\beta\bSigma_{\by_m,1}(k_0)+(1-\beta)\bSigma_{\by_m,2}(k_0)+\bGamma_{2,m}$,
\end{center} where the norm of $\bGamma_{1,m}$ can be made arbitrarily small when $T$ is sufficiently large. By Lemma~\ref{lemma_mkt}, we have 
\begin{align*}
\|\bGamma_{1,m}\|_2 &=\|\frac{1}{k-m}\sum_{t=k_0+1}^{k_0+m}\mathbb E(\by_{t}\by_{t,m}^T)-\frac{m}{k-m}\bSigma_{\by\by_m,2}(k_0)\|_2\\
                   &\leq  \frac{m}{k-m}(\sup_{k_0+1\leq t \leq k_0+m}\|\mathbb E(\by_{t}\by_{t,m}^T)\|_2+\|\bSigma_{\by\by_m,2}(k_0)\|_2)\\
                   &\leq O(\frac{m}{k-m})\\
                   &=O(T^{-1}).
\end{align*}
Similarly, we have $\|\bGamma_{2,m}\|_2 \leq O(T^{-1})$. By Lemma~\ref{sigma_bound}, when $T$ is sufficiently large, we can obtain
\[
\|\bSigma_{\by,1}(k)\|_2\leq\frac{k_0-m}{k-m}\|\bSigma_{\by,1}(k_0)\|_2+\frac{k-k_0}{k-m}\|\bSigma_{\by,2}(k_0)\|_2 \leq \kappa_2 ,
\]
\[
\|\bSigma_{\by\by_m,1}(k)\|_2\leq\beta\|\bSigma_{\by\by_m,1}(k_0)\|_2+(1-\beta)\|\bSigma_{\by\by_m,2}(k_0)\|_2+\|\bGamma_{1,m}\|_2 \leq \kappa_2c_2,
\]
\[
\|\bSigma_{\by_m,1}(k)\|_2\leq\beta\|\bSigma_{\by_m,1}(k_0)\|_2+(1-\beta)\|\bSigma_{\by_m,2}(k_0)\|_2+\|\bGamma_{2,m}\|_2 \leq \kappa_2c_4.
\]
By the definition of the smallest eigenvalue, we have
\begin{align*}
    \lambda_{\min}(\bSigma_{\by,1}(k))&=\min_{\|\bx\|=1}\bx^T\bSigma_{\by,1}(k)\bx \\
    &\geq \min_{\|\bx\|=1}\left(\frac{k_0-m}{k-m}\bx^T\bSigma_{\by,1}(k_0))\bx + \frac{k-k_0}{k-m}\bx^T\bSigma_{\by,2}(k_0))\bx\right)\\
    &= \frac{k_0-m}{k-m}\lambda_{\min}(\bSigma_{\by,1}(k_0)) + \frac{k-k_0}{k-m}\lambda_{\min}(\bSigma_{\by,2}(k_0))\\
    &\geq \kappa_1.
\end{align*}
Similarly, $\lambda_{\min}(\bSigma_{\by_m,1}(k)) \geq \kappa_1c_3$. Following the proof method of Lemma~\ref{lemma_4}, we can similarly complete the proof using the above bound. 

\begin{lemma}\label{lemma_MhatM}
If Assumptions 1-4 hold, for any $k>k_0+\epsilon T$,
\begin{center}
$\mathbb{E}\|\wh\bM_2(k)-\bM_2(k)\|_2 \le C(\kappa_1^{-2}\kappa_2+\kappa_1^{-1})p(T-k)^{-\frac{1}{2}},$\\
$\mathbb{E}\|\wh\bM_1(k)-\bM_1(k)\|_2 \le C(\kappa_1^{-2}\kappa_2+\kappa_1^{-1})p(k-m)^{-\frac{1}{2}}.$
\end{center}
\end{lemma}
\textbf{\proofname.} 
Let $\bR_2=\bSigma_{\by,2}^{-1/2}(k)\bSigma_{\by\by_m,2}(k)\bSigma_{\by_m,2}^{-1/2}(k)$ and $\wh\bR_2=\wh\bSigma_{\by,2}^{-1/2}(k)\wh\bSigma_{\by\by_m,2}(k)\wh\bSigma_{\by_m,2}^{-1/2}(k)$. Then $\bM_2(k)=\bR_2\bR_2{^T}$ and $\wh\bM_2(k)=\wh\bR_2\wh\bR_2{^T}$. We first note that the eigenvalue bounds required for the inverse covariance matrices hold uniformly over the trimmed search set. Indeed, by the weighted representations in Lemmas 4 and 5, the candidate-specific covariance matrices $\bSigma_{\by,i}(k)$ and $\bSigma_{\by_m,i}(k)$ can be expressed as weighted combinations of the regime-specific covariance matrices evaluated at the true change point, up to boundary terms involving at most $m$ observations. Since $m$ is fixed and $k$ belongs to the trimmed search interval, these boundary terms are uniformly negligible. Together with the eigenvalue bounds in Lemma~\ref{sigma_bound}, this implies that, uniformly over admissible $k$ and $i=1,2$,
\[
\lambda_{\min}\{\bSigma_{\by,i}(k)\}\geq C\kappa_1,\qquad
\|\bSigma_{\by,i}(k)\|_2\leq C\kappa_2,
\]
and
\[
\lambda_{\min}\{\bSigma_{\by_m,i}(k)\}\geq C\kappa_1,\qquad
\|\bSigma_{\by_m,i}(k)\|_2\leq C\kappa_2,
\]
for some constant $C>0$ independent of $p,T$ and $k$. Hence the inverse covariance matrices and their inverse square roots are uniformly bounded over the trimmed search set. Combining these uniform bounds with Lemmas 3--5 yields the following perturbation bounds.
\begin{align*}
    \mathbb{E}\|\wh\bR_2-\bR_2\|_2&\leq\|\wh\bSigma_{\by_m,2}^{-\frac{1}{2}}(k)\|\cdot \mathbb{E}\|\hat\bSigma_{\by,2}^{-\frac{1}{2}}(k)\hat\bSigma_{\by\by_m,2}(k)-\bSigma_{\by,2}^{-\frac{1}{2}}(k)\bSigma_{\by\by_m,2}(k)\|\\
    &\quad+\|\bSigma_{\by,2}^{-\frac{1}{2}}(k)\bSigma_{\by\by_m,2}(k)\|\cdot \mathbb{E}\|\wh\bSigma_{\by_m}^{-\frac{1}{2}}(k)-\bSigma_{\by_m}^{-\frac{1}{2}}(k)\|\\
    &\leq\kappa_1^{-\frac{1}{2}}C(\kappa_1^{-\frac{3}{2}}\kappa_2+\kappa_1^{-\frac{1}{2}})p(T-k)^{-\frac{1}{2}}+\kappa_1^{-\frac{1}{2}}\kappa_2C\kappa_1^{-\frac{3}{2}}p(T-k)^{-\frac{1}{2}}\\
    &\leq C(\kappa_1^{-2}\kappa_2+\kappa_1^{-1})p(T-k)^{-\frac{1}{2}}.
\end{align*}
By Assumption~\ref{a4}, we have
\begin{center}
    $\|\bR_2\|_2=\|\bL_1\bSigma_{\bff\bfeta_m}(k)\bSigma_{\bfeta_m}^{-\frac{1}{2}}(k)\|_2 \leq c_3^{-\frac{1}{2}}c_2.$
\end{center}
Therefore, there exists a constant $C$ such that 
\begin{align*}
    \mathbb E\|\wh\bM_2(k)-\bM_2(k)\|_2&\leq \|\wh\bR_2\|\cdot E\|\wh\bR_2-\bR_2 \|_2+\|\bR_2\|\cdot E\|\wh\bR_2-\bR_2\|_2\\
    &\leq C(\kappa_1^{-2}\kappa_2+\kappa_1^{-1})p(T-k)^{-\frac{1}{2}}.
\end{align*}
By the same argument, we have $\mathbb E\|\wh\bM_1(k)-\bM_1(k)\|_2 \le C(\kappa_1^{-2}\kappa_2+\kappa_1^{-1})p(k-m)^{-\frac{1}{2}}.$

\begin{lemma}\label{lemma_lambda_Mbeta}
We define the weighted CCA matrix as $\wt\bM(\beta)=\beta\bM_1(k_0)+(1-\beta)\bM_2(k_0)$, and let $\lambda_j(\widetilde\bM(\beta))$ be the $j$-th largest eigenvalue of $\widetilde \bM(\beta)$, where $\beta \in [(\gamma_1-\epsilon)/\gamma_2,1-\epsilon/\gamma_2]$. If Assumptions 4-6 hold and $k>k_0+\epsilon T$, $\dfrac{m+1}{T}<\epsilon<\gamma_2-\dfrac{k_0}{T}$, $d\leq m < \epsilon T-1$ and $m$ is fixed, then
    \begin{center}
           $\lambda_{r_1+1}(\widetilde \bM(\beta))\geq \epsilon \delta_0^2 c_1^{2}c_4^{-1}\gamma_2^{-1}$, $\lambda_{1}(\widetilde \bM(\beta))\leq c_2^2 c_3^{-1}.$
    \end{center}
\end{lemma}
\textbf{\proofname.} It is straightforward to verify that $\text{rank}(\bM_i(k_0))=\text{rank}(\bL_i^{(1)})=r_i$, and the  canonical correlation matrix $\bM_i(k_0)$ admits the form:

\[
\bM_i(k_0)=\bL_1^{(i)}\bSigma_{\bff\bfeta_m,i}(k_0)\bSigma_{\bfeta_m,i}^{-1}(k_0)\bSigma_{\bfeta_m\bff,i}(k_0)\bL_1^{(i)^T}; i=1,2.
\]

Hence, we obtain $\mathcal{M}(\bM_i(k_0)) \subset \mathcal{M}(\bL_1^{(i)})$. Let $\mathcal{M}(\bL_1^{(1)})^{\perp}$ be the orthogonal complementary space of $\mathcal{M}(\bL_1^{(1)})$ and $\dim(\mathcal{M}(\bL_1^{(1)})^{\perp})=p-r_1$. It is not difficult to show that there exist unit vectors $\bw=\bL_1^{(2)}\bu \in \mathcal{M}(\bL_1^{(2)}), \bv \in \mathcal{M} (\bL_1^{(1)})^{\perp}$ such that $\bv^T\bw=\bu^T\bL_1^{(2)^T}\bv>\delta_0$. Since $\bL_1^{(2)}$ is a column orthogonal matrix, we can show that $\|\bu\|=1$ and $\|\bL_1^{(2)^T}\bv\|\geq\bv^T\bw>\delta_0$. When Assumption~\ref{a4} holds, we can obtain
\begin{align*}
     \bv^T\bM_2(k_0)\bv&=\bv^T\bL_1^{(2)}\bSigma_{\bff\bfeta_m,2}(k_0)\bSigma_{\bfeta_m,2}^{-1}(k_0)\bSigma_{\bfeta_m\bff,2}(k_0)\bL_1^{(2)^T}\bv\\
     &\geq \lambda_{\min}(\bSigma_{\bff\bfeta_m,2}(k_0)\bSigma_{\bfeta_m,2}^{-1}(k_0)\bSigma_{\bfeta_m\bff,2}(k_0))\|\bL_1^{(2)^T}\bv\|^2 \\
     &\geq c_1^2c_4^{-1}\delta_0^2 .
\end{align*}
By the Courant-Fischer theorem, we can construct a Rayleigh quotient of $\widetilde\bM(\beta)$ to compute $\lambda_{r_1+1}(\widetilde\bM(\beta))$:
\begin{align*}
    \lambda_{r_1+1}(\widetilde\bM(\beta))&=\min_{\dim(\mathcal{S})=r_1+1}\max_{\substack{\bx\in\mathcal{S}\\\|\bx\|=1}}\bx^T\widetilde\bM(\beta)\bx\\
    &\geq \min_{\substack{\dim(\mathcal{S})=r_1+1 \\ \bv\in \mathcal{S}\cap \mathcal{M} (\bL_1^{(1)})^{\perp} \\ \|\bv\|=1
    }}\beta\bv^T\bM_1(k_0)\bv+(1-\beta)\bv^T\bM_2(k_0)\bv\\
    &\geq (1-\beta)c_1^2c_4^{-1}\delta_0^2\\
    &\geq \epsilon \delta_0^2 c_1^{2}c_4^{-1}\gamma_2^{-1}.
\end{align*}
For the maximum eigenvalue of $\widetilde\bM(\beta)$, we can obtain by Assumption~\ref{a4},
\begin{center}
    $\lambda_1(\widetilde\bM(\beta))=\|\widetilde\bM(\beta)\|_2\leq \beta\|\bM_1(k_0)\|_2+(1-\beta)\|\bM_2(k_0)\|_2\leq c_2^2 c_3^{-1}$ .    
\end{center}
This lemma shows that the convex combination of $\bM_1(k_0)$ and $\bM_2(k_0)$ can be bounded under the stated conditions. We next turn to the eigenvalue bounds for the mixed matrix $\bM_1(k)$.

\begin{lemma}\label{lemma_lambda(k)_bound}
If Assumptions 2-6 hold, then
    \begin{center}
      $\dfrac{1}{4}\epsilon\delta_0^2 c_1^{2}c_4^{-1}\gamma_2^{-1}\leq \lambda_{r_1+1}^{(1)}(k) \leq \lambda_{1}^{(1)}(k) \leq c_2^2c_3^{-1}+\dfrac{1}{2}\epsilon \delta_0^2 c_1^{2}c_4^{-1}\gamma_2^{-1},$
    \end{center}
where $k>k_0+\epsilon T$, $\dfrac{m+1}{T}<\epsilon<\gamma_2-\dfrac{k_0}{T}$ and $d\leq m < \epsilon T-1$.
\end{lemma}

\textbf{\proofname.} To obtain a rigorous lower bound for $\lambda_{r_1+1}^{(1)}(k)$, we work on the subspace constructed in Lemma~\ref{lemma_lambda_Mbeta} and control only the boundary correction terms, which are of order $O(T^{-1})$. By Lemma~\ref{lemma_lambda_Mbeta}, there exists an $(r_1+1)$ dimensional subspace $\mathcal{U} \subset \mathcal{M}(\bL_1^{(1)})^{\perp}$ such that 
\begin{align*}
    \min_{\substack{\bu\in \mathcal{U}\\ \|\bu\|=1}} \bu^T  \widetilde\bM(\beta) \bu\geq \frac{1}{2}\epsilon \delta_0^2 c_1^{2}c_4^{-1}\gamma_2^{-1}.
\end{align*}
For $k\geq k_0 +\epsilon T$, write
\begin{center}
    $\bM_1(k) = \widetilde\bM(\beta) + \bR_T(k)$,
\end{center}
where $\bR_T(k)$ collects the boundary correction terms induced by $\bGamma_{1,m}$ and $\bGamma_{2,m}$. Since $\|\bGamma_{1,m}\|_2 + \|\bGamma_{2,m}\|_2 \leq O(T^{-1})$, we have
\begin{align*}
    \sup_{\substack{\bu\in \mathcal{U}\\ \|\bu\|=1}} |\bu^T  \bR_T(k) \bu| \leq CT^{-1}.
\end{align*}
Therefore, we have
\begin{align*}
    \min_{\substack{\bu\in \mathcal{U}\\ \|\bu\|=1}} \bu^T  \bM_1(k) \bu\geq \frac{1}{2}\epsilon \delta_0^2 c_1^{2}c_4^{-1}\gamma_2^{-1} -CT^{-1}.
\end{align*}
By the Courant–Fischer theorem,
\begin{align*}
    \lambda_{r_1+1}(\bM_1(k)) \geq  \frac{1}{2}\epsilon \delta_0^2 c_1^{2}c_4^{-1}\gamma_2^{-1} -CT^{-1}.
\end{align*}
If $T^{-1}=o(\varepsilon)$, then, for sufficiently large $T$, we have $\lambda_{r_1+1}(\bM_1(k)) \geq  \dfrac{1}{4}\epsilon \delta_0^2 c_1^{2}c_4^{-1}\gamma_2^{-1}$.

Let $\bA=\bSigma_{\by,1}(k), \bB=\bSigma_{\by\by_m,1}(k)-\bGamma_1(m), \bC=\bSigma_{\by_m,1}(k)-\bGamma_2(m)$, then we define the function $f(\cdot)$ as:
\begin{center}
    $f(\bH_1,\bH_2,\bH_3)=\bH_1^{-1/2}\bH_2\bH_3^{-1}\bH_2^{T}\bH_1^{-1/2}.$
\end{center}
where $\bH_1$ and $\bH_3$ are reversible. By construction, $\widetilde{\bM}(\beta)=f(\bA,\bB,\bC)$, and $\bM_1(k)=f(\bA,\bB+\bGamma_{1,m},\bC+\bGamma_{2,m}).$

Using the perturbation theory of matrix inversion, we have
\begin{align*}
    \|\bM_1(k)-\widetilde{\bM}(\beta)\|_2
    &=
    \|f(\bA,\bB+\bGamma_{1,m},\bC+\bGamma_{2,m})-f(\bA,\bB,\bC)\|_2 \\
    &\leq
    C(\kappa_2\kappa_1^{-2}+\kappa_2^2\kappa_1^{-3})
    \left(\|\bGamma_{1,m}\|_2+\|\bGamma_{2,m}\|_2\right) \\
    &\leq
    C(\kappa_2\kappa_1^{-2}+\kappa_2^2\kappa_1^{-3})T^{-1}.
\end{align*}
By Weyl's inequality, and Lemma~\ref{lemma_lambda_Mbeta},
\begin{align*}
    \lambda_{1}^{(1)}(k)
    &\leq
    \lambda_1(\widetilde{\bM}(\beta))
    +
    \|\bM_1(k)-\widetilde{\bM}(\beta)\|_2\\
    &\leq
    c_2^2c_3^{-1}
    +
    C(\kappa_2\kappa_1^{-2}+\kappa_2^2\kappa_1^{-3})T^{-1}.
\end{align*}
Since $(\kappa_2\kappa_1^{-2}+\kappa_2^2\kappa_1^{-3})T^{-1}=o(\epsilon)$,
we further obtain, for sufficiently large $T$,
\begin{align*}
    \lambda_{1}^{(1)}(k)
    \leq
    c_2^2c_3^{-1}
    +
    \frac{1}{2}\epsilon\delta_0^2 c_1^{2}c_4^{-1}\gamma_2^{-1}.
\end{align*}
Combining the lower and upper bounds completes the proof.

\begin{lemma}\label{lemma_errorG}
Let Assumptions 1-6 hold. Then
\begin{center}
    $(i)\quad \mathbb{E}|\wh G(k)-G(k)|\leq C(\epsilon^{-1}\kappa_1^{-2}\kappa_2+\kappa_1^{-4}\kappa_2^3)pT^{-1/2},$\\
    $(ii)\quad \mathbb{E}|\wh G(k_0)-G(k_0)|\leq C\kappa_1^{-4}\kappa_2^3pT^{-1/2},$   
\end{center}
where $G(\cdot)$ is the criterion function defined with the $L_2$ norm.
\end{lemma}

\textbf{\proofname.} By Lemma~\ref{lemma_mkt}, Lemma~\ref{lemma_MhatM}, Lemma~\ref{lemma_lambda(k)_bound} and Weyl's inequality,  when $k\geq k_0+\epsilon T$ we have
\begin{align*}
    \mathbb{E}|\wh G(k)-G(k)| &=\mathbb{E}|\sum_{i=1}^2 \frac{\wh\lambda_{r_i+1}^{(i)}(k)}{\wh\lambda_1^{(i)}(k)}-\sum_{i=1}^2\frac{\lambda_{r_i+1}^{(i)}(k)}
    {\lambda_1^{(i)}(k)}|\\
    &\leq \mathbb{E}\sum_{i=1}^2 |\frac{\wh\lambda_{r_i+1}^{(i)}(k)}{\wh\lambda_1^{(i)}(k)}-\frac{\lambda_{r_i+1}^{(i)}(k)}{\lambda_1^{(i)}(k)}|\\
    &=\mathbb{E}\sum_{i=1}^2 |\frac{\wh\lambda_{r_i+1}^{(i)}(k)}{\wh\lambda_1^{(i)}(k)}-\frac{\wh\lambda_{r_i+1}^{(i)}(k)}{\lambda_1^{(i)}(k)}+\frac{\wh\lambda_{r_i+1}^{(i)}(k)}{\lambda_1^{(i)}(k)}-\frac{\lambda_{r_i+1}^{(i)}(k)}{\lambda_1^{(i)}(k)}|\\
    &\leq \mathbb{E}\sum_{i=1}^2\left(|\frac{\wh\lambda_{r_i+1}^{(i)}(k)}{\wh\lambda_1^{(i)}(k)}|\cdot\frac{|\wh\lambda_1^{(i)}(k)-\lambda_1^{(i)}(k)|}{\lambda_1^{(i)}(k)}+\frac{|\wh\lambda_{r_i+1}^{(i)}(k)-\lambda_{r_i+1}^{(i)}(k)|}{\lambda_1^{(i)}(k)}\right)\\
    &\leq \mathbb{E}\sum_{i=1}^2\frac{|\wh\lambda_1^{(i)}(k)-\lambda_1^{(i)}(k)|+|\wh\lambda_{r_i+1}^{(i)}(k)-\lambda_{r_i+1}^{(i)}(k)|}{\lambda_1^{(i)}(k)}\\
    &\leq 2\left[ \frac{\mathbb{E}\|\wh\bM_1(k)-\bM_1(k)\|_2}{\frac{1}{4}\epsilon \delta_0^2 c_1^{2}c_4^{-1}\gamma_2^{-1}}+\frac{\mathbb{E}\|\wh\bM_2(k)-\bM_2(k)\|_2}{\kappa_1^2\kappa_2^{-2}}\right]\\
    &\leq 2\left[ \frac{C(\kappa_1^{-2}\kappa_2+\kappa_1^{-1})p(T-k)^{-1/2}}{\frac{1}{4}\epsilon \delta_0^2 c_1^{2}c_4^{-1}\gamma_2^{-1}}+\dfrac{C(\kappa_1^{-2}\kappa_2+\kappa_1^{-1})p(k-m)^{-1/2}}{\kappa_1^2\kappa_2^{-2}} \right]\\
    &\leq C(\epsilon^{-1}\kappa_1^{-2}\kappa_2+\kappa_1^{-4}\kappa_2^3)pT^{-1/2}.
\end{align*}
Similarly, there exists a constant $C$ such that $\mathbb{E}|\wh G(k_0)-G(k_0)|\leq C\kappa_1^{-4}\kappa_2^3pT^{-1/2}$.

\textbf{Proof} \textbf{of} \textbf{Theorem~\ref{tm1}.} We prove the result only for diverging $p$, since the fixed-$p$ case is a special case. Since $G(\wh k)\geq 0$ and $G(k_0)=0$, for any fixed $(m+1)/T<\epsilon<\gamma_2-k_0/T$, let $h(\epsilon)=\epsilon \delta_0^2 c_1^{2}c_4^{-1}\gamma_2^{-1}/(2c_2^2 c_3^{-1}+\epsilon \delta_0^2 c_1^{2}c_4^{-1}\gamma_2^{-1})$, then
\begin{align*}
    \mathbb{P}\{\wh k \geq k_0+\epsilon T\} &= \mathbb{P}\{\wh k \geq k_0+\epsilon T,\wh G(k_0)>\wh G(\wh k)\}\\
    &=\mathbb{P}\{\wh G(k_0)-G(k_0)>\wh G(\wh k)-G(\wh k)+G(\wh k),\wh k \geq k_0+\epsilon T\}\\
    &=\mathbb{P}\{\wh G(k_0)-G(k_0)+ G(\wh k)-\wh G(\wh k)+\frac{3}{4}h(\epsilon)-G(\wh k)>\frac{3}{4}h(\epsilon),\wh k \geq k_0+\epsilon T\}\\
    &\leq \mathbb{P}\{|\wh G(k_0)-G(k_0)|>\frac{1}{4}h(\epsilon)\}+\mathbb{P}\{|\wh G(\wh k)-G(\wh k)|>\frac{1}{4}h(\epsilon),\wh k \geq k_0+\epsilon T\}\\
    &\quad +\mathbb{P}\{G(\wh k)<\frac{1}{2}h(\epsilon),\wh k \geq k_0+\epsilon T\}\\
    &=I_1+I_2+I_3.
\end{align*}
By Lemma~\ref{lemma_lambda(k)_bound}, we can easily obtain $G(\wh k)\geq h(\epsilon)$ under the condition $\wh k \geq k_0+\epsilon T$. Hence $I_3=0$. By Lemma~\ref{lemma_errorG} and Markov inequality, we have 
\begin{center}
    $I_1 \leq C\epsilon^{-1}\kappa_1^{-4}\kappa_2^3 pT^{-1/2}$,
    $I_2\leq C(\epsilon^{-2}\kappa_1^{-2}\kappa_2+\epsilon^{-1}\kappa_1^{-4}\kappa_2^3)pT^{-1/2}.$
\end{center}
Hence, under $p=o\{\min(T^{1/2},\epsilon^{2}\kappa_1^{2}\kappa_2^{-1}+\epsilon\kappa_1^{4}\kappa_2^{-3})\}$,
\begin{center}
    $\mathbb{P}\{\wh k \geq k_0+\epsilon T\}\leq C(\epsilon^{-2}\kappa_1^{-2}\kappa_2+\epsilon^{-1}\kappa_1^{-4}\kappa_2^3)pT^{-1/2}$.
\end{center}
This completes the proof of Theorem~\ref{tm1}. The proof for $\mathbb{P}\{\wh k < k_0-\epsilon T\}$ is analogous. 

\textbf{Proof} \textbf{of} \textbf{Theorem~\ref{tm2}.} By Lemma 3 of \cite{LamYaoBathia_Biometrika_2011} and Lemma~\ref{lemma_MhatM} above, if $p=o\{\min(T^{1/2},\kappa_2^{-1}\kappa_1^2T^{1/2})\}$, 

\begin{center}
    $\mathbb{E}\|\wh\bL_1^{(i)}(k_0)-\bL_1^{(i)}\|_2\leq\dfrac{\mathbb{E}\|\wh\bM_1(k_0)-\bM_1(k_0)\|_2}{\lambda_{r_1}^{(1)}(k_0)}\leq C\kappa_1^{-2}\kappa_2 pT^{-1/2},\quad i=1,2.$
\end{center}
Then we have $\|\wh\bL_1^{(i)}(k_0)-\bL_1^{(i)}\|_2=O_p(\kappa_1^{-2}\kappa_2 pT^{-1/2})$. The proof for $\|\wh\bL_2^{(i)}(k_0)-\bL_2^{(i)}\|_2$ is analogous.
When $p$ is fixed, we similarly obtain $\|\wh\bL_2^{(i)}(k_0)-\bL_2^{(i)}\|_2=O_p(T^{-1/2})$, for any $i=1,2$. This completes the proof.

\textbf{Proof} \textbf{of} \textbf{Theorem~\ref{tm3}.} Let $\widehat{\bf L}_{1}^{(i)}=\widehat{\bf L}_{1}^{(i)}(k_0)$, then for any $i=1,2$ 
\begin{align*}
    \widetilde{D}^2(\mathcal{M}(\widehat{\bf L}_{1}^{(i)}),\mathcal{M}({\bf L}_{1}^{(i)}))&=\dfrac{1}{r_i}\tr\left\{\wh\bL_1^{(i)^T}\left (\wh\bL_1^{(i)}\wh\bL_1^{(i)^T}-\bL_1^{(i)}\bL_1^{(i)^T}\right )\wh\bL_1^{(i)}\right\}\\
    &\leq \|\wh\bL_1^{(i)^T}\left (\wh\bL_1^{(i)}\wh\bL_1^{(i)^T}-\bL_1^{(i)}\bL_1^{(i)^T}\right )\wh\bL_1^{(i)}\|_{2}.
\end{align*}
Note that
\begin{align*}
   \wh\bL_1^{(i)^T}\left (\wh\bL_1^{(i)}\wh\bL_1^{(i)^T}-\bL_1^{(i)}\bL_1^{(i)^T}\right )\wh\bL_1^{(i)}(k_0)&=-\wh\bL_1^{(i)^T}\left (\wh\bL_1^{(i)}-\bL_1^{(i)}\right) \left(\wh\bL_1^{(i)}-\bL_1^{(i)}\right)^T \wh\bL_1^{(i)}\\
   &\quad + \left(\wh\bL_1^{(i)}-\bL_1^{(i)}\right)^T \left(\wh\bL_1^{(i)}-\bL_1^{(i)}\right),
\end{align*}
which implies
\begin{center}
    $\widetilde{D}(\mathcal{M}(\widehat{\bf L}_{1}^{(i)}),\mathcal{M}({\bf L}_{1}^{(i)})) \le \sqrt{2}\|\wh\bL_1^{(i)}(k_0)-\bL_1^{(i)}\|_2.$
\end{center}
The conclusion follows from the above inequality and Theorem~\ref{tm2}. This completes the proof.



\begin{thebibliography}{2}

    \bibitem[Bai, 2003]{bai2003}
Bai, J. (2003). Inferential theory for factor models of large dimensions. {\em Econometrica}, {\bf 71(1)}: 135--171.
    \bibitem[Bai { et al.}(2020)]{BaiHanShi_2020}
Bai, J., Han, X., Shi, Y. (2020). Estimation and Inference of Change Points in High Dimensional Factor Models. {\em Journal of Econometrics} {\bf 219(1)}: 66--100.
    \bibitem[{Bai and Ng(2002)}]{BaiNg_Econometrica_2002}
Bai, J., Ng, S. (2002). Determining the number of factors in approximate factor models. {\em Econometrica} {\bf 70}: 191--221.
    \bibitem[{Box and Tiao(1977)}]{BoxTiao_1977}
Box, G. E. P., Tiao, G. C. (1977). A canonical analysis of multiple time series. {\em Biometrika} {\bf 64}: 355--365.
    \bibitem[Breitung and Eickmeier(2011)]{BreEick_2011}
Breitung, J., Eickmeier, S. (2011). Testing for structural changes in dynamic factor models. {\em Journal of Econometrics} {\bf 163(1)}: 71--84.
    \bibitem[Chamberlain and Rothschild (1983)]{ChamberainRothschild_1983}
Chamberlain, G.,  Rothschild, M. (1983). Arbitrage, Factor Structure, and Mean-Variance Analysis on Large Asset Markets. {\em Econometrica} {\bf 51(5)}: 1281--1304.
    \bibitem[Chang~et~al., 2015]{changguoyao2015}
Chang, J., Guo, B., Yao, Q. (2015).
\newblock High dimensional stochastic regression with latent factors, endogeneity and nonlinearity. \newblock {\em Journal of Econometrics} {\bf 189}(2): 297--312.
    \bibitem[Chen {et al.}(2014)]{Chenetal_2014}
Chen, L., Dolado, J. J., Gonzalo, J. (2014). Detecting big structural changes in large factor models. {\em Journal of Econometrics} {\bf 180}: 30--48.
    \bibitem[{Davis {et al.}(2016)}]{Davis2016}
Davis, R. A., Zang, P., Zheng, T. (2016). Sparse vector autoregressive modeling. {\em Journal of Computational
and Graphical Statistics} {\bf25(4)}: 1077--1096.
    \bibitem[Davydov (1968)]{Davydov_1968} Davydov, Yu, A. (1968). Convergence of Distributions Generated by Stationary Stochastic Processes. {\em Theory of Probability and Its Applications} {\bf 13(4)}: 691--696. 

    \bibitem[{Gao {et al.}(2018)}]{gaoetal2017}
Gao, Z., Ma, Y., Wang, H., Yao, Q. (2019). 
\newblock Banded spatio-temporal autoregressions. 
\newblock {\em Journal of Econometrics} {\bf 208}: 211--230. 
    \bibitem[Gao and Tsay(2019)]{gaotsay_2019}
Gao, Z., Tsay, R. S. (2019). A Structural-Factor Approach to Modeling High-Dimensional Time Series and Space-Time Data. {\em Journal of Time Series Analysis} {\bf 40(3)}: 343–-362. 
    \bibitem[{Han { et al.}(2015)}]{HanLuLiu_2015}
Han, F., Lu, H. and Liu, H. (2015). A direct estimation of high dimensional stationary vector autoregressions. {\em Journal of Machine Learning Research} {\bf 16(1)}: 3115--3150.
    \bibitem[{Lam {et al.}(2011)}]{LamYaoBathia_Biometrika_2011}
Lam, C., Yao, Q., Bathia, N. (2011). Estimation of latent factors for high-dimensional time series. {\em Biometrika} {\bf98}: 901--918.
    \bibitem[Lam and Yao, 2012]{lamyao2012}
Lam, C., Yao, Q. (2012).
\newblock Factor modeling for high-dimensional time series: inference for the number of factors.
\newblock {\em The Annals of Statistics} {\bf 40}(2): 694--726.
    \bibitem[{Liu and Chen(2016)}]{LiuChen_D_2016}
Liu, X., Chen, R., (2016). Regime-switching factor models for high-dimensional time series. {\em Statistica Sinica} {\bf 26}: 1427--1451.
    \bibitem[{Liu and Chen(2020)}]{LiuChen_D_2020}
Liu, X., Chen, R. (2020). Threshold factor models for high-dimensional time series. {\em Journal of Econometrics} {\bf216}: 53--70.     
    \bibitem[Liu and Zhang(2022)]{LiuZhang_2022}
Liu, X., Zhang, T. (2022). Estimating Change-Point Latent Factor Models for High-Dimensional Time Series. {\em Journal of Statistical Planning and Inference} {\bf 217}: 69--91.
    \bibitem[Ma and Su (2018)]{MaSu_2018}
Ma, S., Su, L. (2018). Estimation of large dimensional factor models with an unknown number of change points. {\em Journal of Econometrics} {\bf 207}: 1--29.
    \bibitem[Onatski, 2010]{Onatski_2010}
Onatski, A. (2010). Determining the Number of Factors from Empirical Distribution of Eigenvalues. {\em The Review of Economics and Statistics} {\bf 92(4)}: 1004--1016.
    \bibitem[Onatski, 2012]{Onatski_2012}
Onatski, A. (2012). Asymptotics of the Principal Components Estimator of Large Factor Models with Weakly Influential Factors. {\em Journal of Econometrics} {\bf 168(2)}: 244--258.    
    \bibitem[Pan and Yao, 2008]{panyao2008}
Pan, J.,  Yao, Q. (2008).
\newblock Modelling multiple time series via common factors.
\newblock {\em Biometrika} {\bf 95}(2): 365--379.
    \bibitem[Pelger, 2020]{Pelger_2020}
Pelger, M. (2020). Understanding Systematic Risk: A High-Frequency Approach. {\em Journal of Finance} {\bf 75(4)}: 2179--2220.

    \bibitem[Schmitt, 1992]{schmitt1992}
Schmitt, B. A. (1992). Perturbation bounds for matrix square roots and Pythagorean sums. {\em Linear Algebra and its Applications}, {\bf 174}, 215--227.
    \bibitem[{Shojaie and Michailidis(2010)}]{ShojaieMichailidis_2010}
Shojaie, A., Michailidis, G. (2010). Discovering graphical Granger causality using the truncated lasso penalty. {\em Bioinformatics} {\bf 26}: 517--523.
    \bibitem[{Stock and Watson(2002)}]{StockWatson_2002}
Stock, J. H., Watson, M. W. (2002). Forecasting using principal components from a large number of predictors. {\em Journal of the American Statistical Association}  {\bf 97}: 1167--1179.
    \bibitem[{Stock and Watson(2005)}]{StockWatson_2005}
Stock, J. H., Watson, M. W. (2005). Implications of dynamic factor models for VAR analysis. NBER Working Paper No. 11467. Available at {\em www.nber.org/papers/w11467}.
    \bibitem[{Tiao and Tsay(1989)}]{TiaoTsay_1989}
Tiao, G. C.,  Tsay, R. S. (1989). Model specification in multivariate time series (with discussion). 
\newblock {\em Journal of the Royal Statistical Society: Series B (Statistical Methodology)} {\bf51}: 157--213.
    \bibitem[Wu, 2016]{Wu_2016}
Wu, J. (2016). Robust Determination for the Number of Common Factors in the Approximate Factor Models. {\em Economics Letters} {\bf 144}: 102--106.
    \bibitem[{Xia {et al.}(2017)}]{Xia_TCR_2017}
Xia, Q., Liang, R., Wu, J. (2017). Transformed Contribution Ratio Test for the Number of Factors in Static Approximate Factor Models. {\em Computational Statistics and Data Analysis} {\bf112}: 235--241.


\end{thebibliography}
\end{document}